\documentclass{article}
\usepackage[a4paper]{geometry}

\usepackage{tikz}
\usetikzlibrary{matrix,calc,trees}
\tikzset{my arrow/.style={blue!60!black,-latex},
  set@com@col/.style={},set@com@col@aryarg/.style={column #1/.style={set@com@col}},
  set@com@row/.style={},set@com@row@aryarg/.style={row #1/.style={set@com@row}},
  set common column/.style 2 args={set@com@col/.style={#2}, set@com@col@aryarg/.list={#1}},
  set common row/.style 2 args={set@com@row/.style={#2}, set@com@row@aryarg/.list={#1}},
}

\usepackage{float}
\usepackage{graphicx}
\usepackage{balance} 
\usepackage{multirow}
\usepackage{booktabs}
\usepackage{hyperref}
\usepackage{nicefrac}
\usepackage{enumitem}

\title{Faster motif counting via succinct color coding\linebreak and adaptive sampling}

\usepackage{bbm}
\usepackage{amsmath}
\usepackage{amssymb}
\usepackage{algorithm}
\usepackage[noend]{algpseudocode}
\usepackage{relsize}
\usepackage{xspace}

\usepackage[bottom]{footmisc}

\usepackage{amsthm}
\newtheorem{theorem}{Theorem}
\newtheorem{lemma}{Lemma}

\author{
Marco Bressan\\
Universit\`a degli Studi di Milano\\
\texttt{marco.bressan@unimi.it}
\and
Stefano Leucci\\
Universit\`a dell'Aquila\\
\texttt{stefano.leucci@univaq.it}
\and
Alessandro Panconesi\\
Sapienza Universit\`a di Roma\\
\texttt{ale@di.uniroma1.it}
} 

\date{}

\begin{document}

\newtheorem{hypothesis}{Hypothesis}

\newcommand{\dsstyle}[1]{\textsc{\small{#1}}}
\newcommand{\infroadusa}{\dsstyle{Road-US}}
\newcommand{\facebook}{\dsstyle{Facebook}}
\newcommand{\orkut}{\dsstyle{Orkut}}
\newcommand{\yeast}{\dsstyle{YeastProtein}}
\newcommand{\yelp}{\dsstyle{Yelp}}
\newcommand{\hollywood}{\dsstyle{Hollywood}}
\newcommand{\lj}{\dsstyle{LiveJournal}}
\newcommand{\wordassoc}{\dsstyle{WordAssoc}}
\newcommand{\twitter}{\dsstyle{Twitter}}
\newcommand{\skitter}{\dsstyle{Skitter}}
\newcommand{\patents}{\dsstyle{Patents}}
\newcommand{\amazon}{\dsstyle{Amazon}}
\newcommand{\fs}{\dsstyle{Friendster}}
\newcommand{\dblp}{\dsstyle{Dblp}}
\newcommand{\roadca}{\dsstyle{Road-CA}}
\newcommand{\roadpa}{\dsstyle{Road-PA}}
\newcommand{\bstan}{\dsstyle{Berk\-Stan}}
\newcommand{\LAW}{\small{LAW}}
\newcommand{\SWS}{\small{MPI-SWS}}
\newcommand{\NDR}{\small{NDR}}
\newcommand{\SNAP}{\small{SNAP}}
\newcommand{\YLP}{\small{YLP}}
\newcommand{\prob}{\operatorname{Pr}}
\newcommand{\bfx}{\mathbf{x}}
\newcommand{\bfc}{\mathbf{c}}
\newcommand{\bfy}{\mathbf{y}}
\newcommand{\E}{\mathbb{E}}
\newcommand{\gk}{\mathcal{G}_k}
\newcommand{\vk}{\mathcal{V}}
\newcommand{\ek}{\mathcal{E}}
\newcommand{\gnk}{\mathcal{G}}
\newcommand{\var}[1]{\operatorname{Var}[#1]}
\newcommand{\cov}{\operatorname{Cov}}
\newcommand{\cerr}{\bar{c}}
\newcommand{\glet}{\ensuremath{x}}
\newcommand{\motivo}{\textsc{Motivo}}
\newcommand{\motiveight}{\textsc{L8Motif}}
\newcommand{\cc}{\textsc{cc}}
\newcommand{\Rplus}{\protect\hspace{-.1em}\protect\raisebox{.35ex}{\smaller{\smaller\textbf{+}}}}
\newcommand{\Cpp}{\mbox{C\Rplus\Rplus}\xspace}
\newcommand{\labelspacing}{6pt}
\hyphenation{graph-let}
\hyphenation{graph-lets}
\newcommand{\sample}{\ensuremath{\operatorname{sample}}}

\newcommand{\tc}{T_C}
\newcommand{\tca}{T'_{C'}}
\newcommand{\tcb}{T''_{C''}}
\newcommand{\ta}{T'}
\newcommand{\tb}{T''}
\newcommand{\transpose}{{\mathsmaller T}}

\maketitle

\begin{abstract}
We address the problem of computing the distribution of induced connected subgraphs, aka \emph{graphlets} or \emph{motifs}, in large graphs.
The current state-of-the-art algorithms estimate the motif counts via uniform sampling by leveraging the color coding technique by Alon, Yuster and Zwick.
In this work we extend the applicability of this approach by introducing a set of algorithmic optimizations and techniques that reduce the running time and space usage of color coding and improve the accuracy of the counts.
To this end, we first show how to optimize color coding to efficiently build a compact table of a representative subsample of all graphlets in the input graph.
For $8$-node motifs, we can build such a table in one hour for a graph with $65$M nodes and $1.8$B edges, which is $2000$ times larger than the state of the art.
We then introduce a novel adaptive sampling scheme that breaks the ``additive error barrier'' of uniform sampling, guaranteeing multiplicative approximations instead of just additive ones.
This allows us to count not only the most frequent motifs, but also extremely rare ones.
For instance, on one graph we accurately count nearly $10.000$ distinct $8$-node motifs whose relative frequency is so small that uniform sampling would literally take centuries to find them.
Our results show that color coding is still the most promising approach to scalable motif counting.
\end{abstract}

\section{Introduction}
Counting the number of copies of a given pattern graph in a host graph is one of the basic graph mining primitives, with applications in network analysis~\cite{Abdelzaher&2015}, graph classification~\cite{Yaveroglu&2014}, graph clustering~\cite{Yin&2017}, and biology~\cite{Yeger-Lotem5934}.
Of particular interest are the subgraphs that are induced and connected, which are commonly known as \emph{graphlets} or \emph{motifs}.
Indeed, motifs are often seen as ``high-order edges'' that are the true building blocks of real-world networks and give fundamental insights into the nature of a graph~\cite{Jha&2015,interactome,Yaveroglu&2014,Yin&2017}.
The problem of counting the number of copies of a given motif in a graph has a long and rich history, which started with triangle counting and evolved towards larger and more complex motifs~\cite{Ahmed&2015,Bhuiyan&2012,Chakaravarthy&2016,Chen&2016,Han&2016,Jha&2015,Pinar&2017,Slota&2013,Wang&2014,Wang&2015,Wang&2016,Zhao&2010}.
Entire frameworks have been designed to make motif mining easy, including systems based on graph databases such as Arabesque~\cite{Arabesque} and GraphSig~\cite{graphsig}, or standalone systems such as Fractal~\cite{Fractal} and AutoMine~\cite{Automine}.

Unfortunately, motif counting becomes quickly intractable with the size $k$ of the motif.
For this reason exact counting is practically feasible only for $k \le 5$, save for special cases such as counting cliques in sparse graphs.
This hardness is not surprising, since the problem is widely believed to require time $n^{\Omega(k)}$~\cite{Chen&2006} where $n$ is the number of nodes in the input graph.
The natural approach to overcome this barrier is to abandon exact counting in favor of approximate counting.
Approximate counting can replace exact counting in many cases, such as in hypothesis testing (deciding if a graph comes from a certain distribution or not) or in estimating the clustering coefficient of a graph (the fraction of triangles among $3$-node motifs).
In this work we focus on approximate motif counting, with special attention on guarantees.
More precisely, we aim to estimate as accurately as possible the number of occurrences of \emph{every} possible distinct motif on $k$ nodes (the star, the clique, the path, etc.)\ in a graph.
Formally, suppose we are given a simple graph $G$ (e.g., a social network), an integer $k > 2$, and two approximation parameters $\epsilon,\delta \in (0,1)$.
For each motif $H$ on $k$ nodes (the star, the clique, the path, etc.), we want an estimate of the number of induced copies of $H$ in $G$, so that with probability at least $1-\delta$ all estimates are within a factor $(1 \pm \epsilon)$ of the actual values.
Our goal is to develop practical algorithms that solve this problem for $G$ and $k$ significantly larger than the state of the art.
This means we aim at graphs $G$ with billions of edges and motifs on more than $5$ nodes.
Note that we are looking at \emph{induced} copies; counting non-induced copies can be significantly easier (think of the stars).
We also remark that, following all previous literature, graphlets are defined as connected.

The most natural approach to motif counting is combinatorial counting.
Unfortunately, this approach requires enumerating and/or counting a number of subgraphs that can grow as $n^{\Omega(k)}$, and therefore does not scale to large $G$ and $k$.
Indeed, even state-of-the-art exact counting algorithms such as~\cite{Pinar&2017} or~\cite{Automine} are reported to work only for $k \le 5$.
We note that combinatorial explosion also affects  approximation algorithms as long as one wants to estimate the counts of \emph{all} motifs at once, as we do.
Consider indeed $N_k$, the number of distinct graphlets on $k$ nodes (that is, the number of non-isomorphic connected simple $k$-node graphs).
Obviously, estimating all $k$-graphlet counts takes time $N_k$ since we might need to output one count for each graphlet.
One can show that $N_k$ grows extremely fast, as $N_k=\exp(\Omega(k^2))$; for example, $N_{20} > 10^{30}$, so already for $k=20$ the task is hopeless in practice.\footnote{See sequence A001349 in the OEIS, \url{https://oeis.org/A001349}.}
However, $N_8\simeq 11\,000$ and $N_{10} \simeq 12\,000\,000$, so for these values of $k$ the task could still be feasible even on large graphs such as real-world social networks.

With combinatorial algorithms ruled out, the most appealing approach left is sampling.
The idea is just to sample graphlet copies uniformly at random from $G$ and estimate their frequencies and counts consequently.
The difficulty lies in implementing the graphlet sampling primitive in an efficient way, which is trickier than it may appear.
The first general graphlet sampling technique was based on Markov chains and was introduced in~\cite{Bhuiyan&2012}.
The approach is elegant and works in principle for every $G$ and $k$, but in practice it is efficient only for $k=4,5$ on medium-sized graphs.
It was later proved that this approach needs $n^{\Omega(k)}$ steps in the worst case to produce just a single unbiased graphlet sample~\cite{Bressan&2017}, making it comparable to brute-force enumeration.
On the other hand,~\cite{Bressan&2017} showed that one can efficiently sample subgraphs using the color coding technique of Alon, Yuster and Zwick~\cite{Alon&1995}.
The idea of color coding is to assign to each node of $G$ one of $k$ colors independently and uniformly at random.
Then, via dynamic programming,  one can count the number of subtrees of $G$ that span $k$ distinct colors (we say they are \emph{colorful}) in time $O(E_G)$, which gives an estimate of the actual number of subtrees in $G$.
While~\cite{Alon&1995} used color coding for counting trees,~\cite{Bressan&2017} showed how to use it to implement graphlet sampling.
This sampling framework has two components.
The first component is an enriched version of the aforementioned dynamic programming, which builds an abstract ``urn'' containing a representative sub-population of all the trees of $G$ on up to $k$ nodes (henceforth ``$k$-treelets'').
The second component is a recursive algorithm that uses the urn to sample a random $k$-treelet from $G$ and, thus, obtain a random connected subset of $k$ nodes (that is, a motif).
One can thus estimate graphlet counts in two steps: the \emph{build-up phase}, where one builds the urn from $G$, and the \emph{sampling phase}, where one samples $k$-treelets from the urn.
The build-up phase takes time $O(a^k m)$ and space $O(a^k n)$ for some $a>0$, where $n$ and $m$ are the number of nodes and edges of $G$, while sampling takes a variable but typically small amount of time.
This algorithm, named CC by the authors (after \textit{c}olor \textit{c}oding), can reliably and accurately count motifs on $k>5$ nodes on medium-large graphs and is the current state of the art in motif counting~\cite{Bressan&2018b}.

While CC was the first algorithm consistently able to count motifs on more than $5$ nodes on large graphs, the hardness of the problem still imposed some limitations on it.
First, the build-up phase of CC is resource-demanding, especially concerning the memory usage, which as mentioned above, grows exponentially with $k$.
This limits the scalability; for example, even on a machine with 72GB of main memory, CC runs out of memory while estimating $7$-graphlet counts on graphs with more than $2$M nodes, as shown in our experiments.
Second, since CC samples graphlets uniformly at random, its approximation guarantees are only additive.
That is, using $s$ samples, CC can only detect graphlets whose relative frequency is at least $\nicefrac{1}{s}$ --- all other graphlets will be undetected or heavily misestimated. 
Since in many graphs nearly all graphlets have extremely low frequencies, CC would need to draw an unacceptable number of samples (this is true not only for CC but for any algorithm based on uniform sampling).
In this work we overcome these two bottlenecks, pushing the color-coding approach in the realm of massive graphs.

\subsection{Our results}
We present two motif counting algorithms.
The first is \motivo, which is designed to count motifs on $k \le 16$ nodes.
The second is \motiveight\ (pronounced ``leitmotif'', for \textit{L}arge-graph $8$-node \textit{M}otif counter), which is optimized for motifs on $k \le 8$ nodes.
In particular, \motiveight\ can count motifs on graphs with billions of edges with excellent accuracy, using just ordinary hardware.
To convey the idea, in one hour we can accurately count $8$-node motifs in a graph with $65$M nodes and $1.8$B edges (\fs); this is 200 and 2000 times larger, in terms of $n$ and $m$, than the prior art.\footnote{All measurements are obtained on a workstation with 36 cores and 72GB of main memory --- see Section~\ref{sec:exp}.}
Unlike all previous algorithms, \motivo\ and \motiveight\ compute accurate counts for nearly \emph{all} graphlets at once, even extremely rare ones.
Consider for instance the \yelp\ graph (one of our datasets).
On this graph, the state-of-the-art algorithm CC finds only the top two most frequent $8$-graphlets and literally misses all the others.
In the same amount of time, \motiveight\ produces counts within a multiplicative error of $\epsilon \le 0.25$ for $\simeq 10.000$ distinct graphlets simultaneously.
Many of these graphlets have frequency $\le 10^{-21}$, which means that any algorithm based on uniform sampling, like CC or random walks, would need $10^3$ years to find them even by sampling $10^6$ graphlets/second.
Figure~\ref{fig:scaling} gives a pictorial summary of these results.

\begin{figure}[t]
\includegraphics[height=130pt]{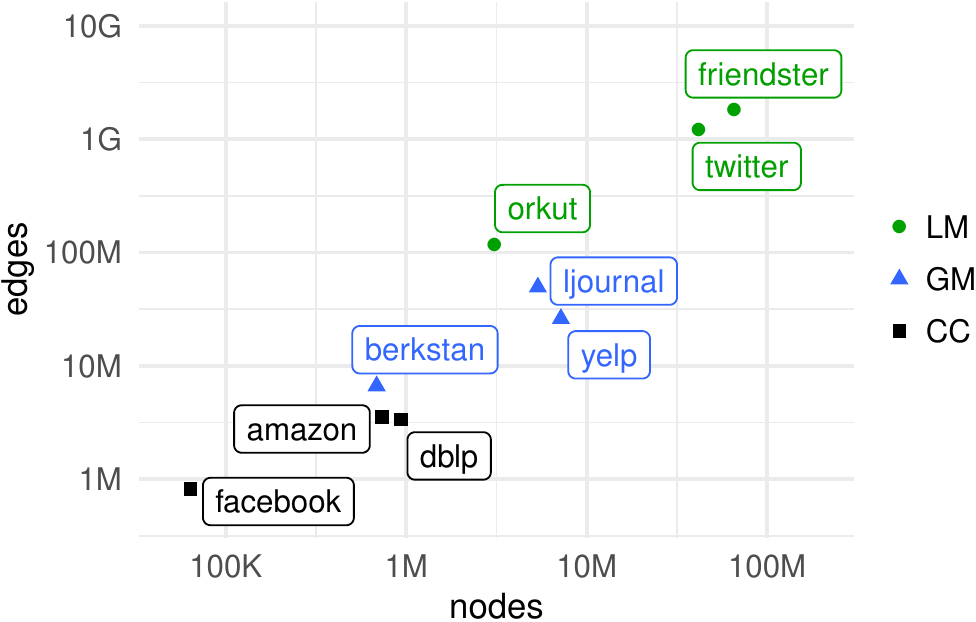}
\hfill
\includegraphics[height=130pt]{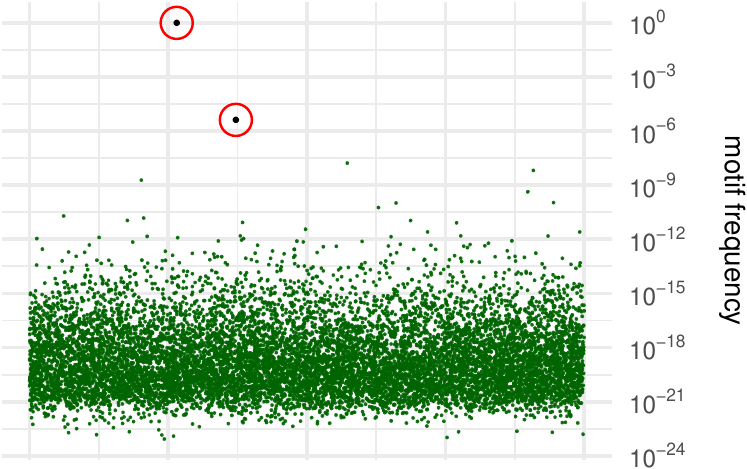}
\caption{Our performance in a picture for motifs on $k=8$ nodes.
Left: the size of the graphs managed by the state-of-the-art CC and by our algorithms \motivo\ and \motiveight\ (abbreviated as LM).
Right: the dense green band has one point for each distinct graphlet that we count accurately on the \yelp\ graph ($\simeq{10.000}$ in total) while the two red circles show the only two graphlets detected by CC.}
\label{fig:scaling}
\end{figure}
\textbf{Technical contributions}.
Our algorithms \motivo\ and \motiveight\ rely on two technical contributions of different nature.
The first contribution is a set of efficient algorithms and data structures that improve the running time and especially the space usage of the build-up phase (one of the main bottlenecks of CC).
In particular, for \motiveight\ we design data structures that use a nearly-optimal amount of bits in an information-theoretical sense, by adopting a compact integer representation of rooted colored treelets, and a variable-length encoding of the treelet counts.
In addition, we show how to entirely skip the heaviest round of the dynamic program via a balanced treelet decomposition trick, saving additional time and space.
Thanks to these ingredients we can scale the build-up phase from millions to billions of edges and from $k=5$ to $k=8$, although from an asymptotic point of view we are still subject to the $O(a^k m)$ running time and $O(c^k n)$ space usage bounds of CC.
Our other algorithm, \motivo, can be used to handle graphlets on up to $k = 16$ nodes (with some loss of efficiency compared to \motiveight).

The second technical contribution, common to both \motivo\ and \motiveight, is for the sampling phase and is of a fundamentally different nature.
To convey the idea, imagine having an urn with 1000 balls --- 990 red, 9 green, and 1 blue.
With uniform sampling, we will quickly obtain a good estimate of the fraction of red balls, but we will need many samples to observe green or blue balls.
This is the uniform sampling barrier mentioned above.
Now imagine to remove ${\sim}99\%$ of all red balls from the urn.
We would be left with 10 red, 9 green, and 1 blue ball, and we would then quickly get a good estimate of the fraction of green balls.
Then imagine deleting ${\sim}99\%$ of the red and green balls; after this, we could quickly estimate the fraction of blue balls.
We show that our treelet database supports a similar ``deletion trick'': we can ignore some treelets (say, stars) and focus on other ones (say, paths).
In this way we can estimate the most frequent graphlet, delete it and proceed to the second most frequent one, and so on.
We name this algorithm Adaptive Graphlet Sampling, or AGS.
Technically speaking, AGS is based on an online greedy algorithm for a fractional set cover problem (we want to ``cover'' all graphlets with their spanning treelets).
We provide theoretical guarantees on the accuracy and efficiency of AGS via competitive analysis and martingale concentration bounds.
We note that AGS is the first algorithm to ensure accuracy for all graphlets at once.

\textbf{Remark.}
We first described \motivo\ in~\cite{Bressan&2019VLDB}, a preliminary version of this work.
The main contribution of this extended version is \motiveight\ and all the associated technical ideas (integer treelet encoding, round skipping, and variable-length counts).
Thanks to these ideas we can reduce the running time and space usage of \motivo\ by almost an order of magnitude, which allows us to scale $k=7$ to $k=8$ on the largest graphs.
Recall that the time and space requirements grow exponentially with $k$, hence, scaling from $k$ to $(k+1)$ is a considerable challenge.
To give an idea, to count $6$-node motifs on our largest graph \fs\ ($\sim 2$B nodes), \motivo\ needs 130 GB of space and 3 hours of time, while \motiveight\ needs only 25GB and 30 minutes.
The source code of both \motivo\ and \motiveight\ is publicly available at \url{https://gitlab.com/steven3k/motivo}.

\paragraph{Manuscript organization.}
Subsections~\ref{sub:rel} and~\ref{sub:prelim} respectively discuss related work and introduce notation and conventions.
Section~\ref{sec:cc} reviews color coding and CC, the starting points of our work.
Section~\ref{sec:speed} describes the build-up phase of our algorithms, comparing them to CC.
Section~\ref{sec:sample} describes our adaptive graphlet sampling technique and other graphlet sampling optimizations.
Finally, Section~\ref{sec:exp} provides experimental results on several public datasets.

\subsection{Related work}
\label{sub:rel}
The traditional approach to subgraph counting is based on combinatorial counting.
This approach works for $k=4$ and $5$, but becomes unusable for $k>5$, even for approximate counting.
This limitation affects many motif mining tools such as Arabesque~\cite{Arabesque}, ESCAPE~\cite{Pinar&2017}, Fractal~\cite{Fractal}, Pangolin~\cite{Pangolin}, AutoMine~\cite{Automine}, and DISC~\cite{Zhang20distributed}.
None of these works reports results for $k>5$.
The only algorithms that can scale to $k>5$ are based on random walks or color coding.

The random walk approach works as follows.
We define a graph $\mathcal{G}_k$ whose nodes are the graphlets of $G$ and where two graphlets are adjacent if they share $k-1$ nodes.
Then, the lazy random walk on $\mathcal{G}_k$ can be shown to be ergodic and converge towards a unique limit distribution on the set of all graphlets of $G$.
Moreover, each step of the walk can be efficiently simulated using only local information from $G$.
Thus one just simulates enough steps, draw a graphlet from the limit distribution, and compute an unbiased estimator of its frequency~\cite{Bhuiyan&2012,Wang&2014,Chen&2016,Han&2016}.
Unfortunately, the random walk may need $\Omega(n^{k-1})$ steps to reach the limit distribution, or even just to find the most frequent graphlet of $G$ when $G$ is fast-mixing~\cite{Bressan&2017,Bressan&2018b}.
The running time of this approach is therefore close to $O(n^{k})$, which is the running time of the naive brute-force enumeration.
One can mitigate this mixing time explosion by walking on subgraphs on less than $k$ nodes~\cite{Chen&2016} or by sampling a spanning tree directly in $G$~\cite{Paramonov2019}.
Unfortunately, this gives biased samples that needs to be reweighted, increasing the estimator variance so that $\Omega(n^{k-1})$ samples become necessary again.
In addition, most of these algorithms can estimate only the relative \emph{frequencies} of graphlets, but not their counts.

There are several algorithms based on color coding.
Most of them can only count special motifs, such as trees and subgraphs of small treewidth.
Moreover, often they count only non-induced copies~\cite{Hffner2007AlgorithmEF,Alon&2008,Zhao&2010,Zhao&hadoop,Slota&2013,Finocchi&2015MapReduce,Chakaravarthy&2016}.
Here instead we want to count the \emph{induced} copies of \emph{all} graphlets at once.
The only algorithm that can do so is the CC algorithm of~\cite{Bressan&2017,Bressan&2018b}.
The only two limitations of CC are that (1) for $k=8$, it can only manage graphs on less than $0.5$M nodes, and (2) it gives accurate counts only for the few most frequent graphlets.
In contrast, thanks to our optimization of the build-up phase and the introduction of adaptive sampling, we manage graphs with tens of millions of nodes and billions of edges while giving guarantees for nearly all $k$-graphlets at once.
We note that specialized algorithms can efficiently compute or estimate the number of cliques on graphs with tens or hundreds of millions of edges \cite{Automine, Jain&2017, DBLP:conf/wsdm/0003S20}.
Again, these algorithms work only for a particular motif (the clique) while we can count all motifs at once.

\subsection{Preliminaries and notation.}
\label{sub:prelim}
We denote the host graph by $G=(V,E)$, and we let $n=|V|$ and $m=|E|$.
A \emph{graphlet} is a connected graph $H=(V_H,E_H)$.
An \emph{occurrence} or \emph{copy} of $H$ in $G$ is a subgraph of $G$ isomorphic to $H$.
Unless otherwise specified, a copy of $H$ in $G$ is meant as induced.
We let $k=|V_H|$; in this work we consider $k \le 16$, and we pay particular attention to $k \le 8$.
A \emph{treelet} $T$ is a graphlet that is a tree.
When using treelets as spanning trees, their copies in $G$ are meant as not necessarily induced.
We denote by $\mathcal{H}=\mathcal{H}_k$ the set of all $k$-node graphlets, i.e., all non-isomorphic connected graphs on $k$ nodes.
When needed we denote by $H_i$ the $i$-th graphlet of $\mathcal{H}$.
A colored graphlet has a color $c_u \in [k]$ associated to each one of its nodes $u$.
With a mild abuse of notation, in this paper we use $[k]$ to denote the set $\{0,\ldots,k-1\}$.
A graphlet is \emph{colorful} if its nodes have pairwise distinct colors.
We denote by $C \subseteq [k]$ a subset of colors.
We denote by $(T,C)$ or $\tc$ a colored treelet whose nodes span the set of colors $C$; we only consider colorful treelets, i.e., the case $|T|$=$|C|$.
We often consider treelets and colored treelets rooted at a node $r \in T$ (different rootings can give different treelets).
Finally, by $d_v$ we denote the degree of $v$ in $G$, and by $u \sim v$ we indicate that $u$ is a neighbor of $v$ in $G$.

\section{Color coding and the CC algorithm}
\label{sec:cc}
\label{sub:cc}
Our algorithms, like the CC algorithm of~\cite{Bressan&2017,Bressan&2018b}, are based on the color coding technique by Alon, Yuster and Zwick~\cite{Alon&1995}, which works as follows.
First, we assign to each node $v \in G$ a color chosen uniformly and independently in $[k]$.
Now consider any $k$-node treelet in $G$: the random coloring makes it colorful with probability $p_k=\frac{k!}{k^k} \approx e^{-k}$.
Therefore a constant fraction of all treelets of $G$ will become colorful.
The main idea of color coding is that the number of colorful $k$-treelets can be counted in time $O(m \cdot a^k)$ with a bottom-up dynamic program.
The running time is exponential in $k$, but linear in $m$.

Color coding was introduced to detect and count noninduced trees.
Then, the authors of~\cite{Bressan&2017,Bressan&2018b} showed how to extend it for sampling graphlets.
The idea is to run a modified dynamic program that collects information about the ``colorful structure'' of $G$.
Once this is done, it is easy to sample colorful $k$-treelets from $G$ and, so, to sample colorful graphlets (just take the graphlet spanned by the treelet).
This is the essence of the CC algorithm~\cite{Bressan&2017,Bressan&2018b}.
We now detail the two phases of CC, the \emph{build-up phase} and the \emph{sampling phase}, which will also be the same phases used in our algorithms.

\subsection{The build-up phase}
\label{sub:build}
The goal of this phase is to build a \emph{count table} holding the counts of colorful treelets of $G$.
The phase starts by coloring $G$: for each node $v$, we draw a color $c_v$ uniformly at random in $[k]$.
Now, for every $v \in G$ and every rooted colored treelet $\tc$ on up to $k$ nodes, we want to compute the following quantity:
\begin{align}
c(\tc,v) = \text{the number of copies of } \tc \text{ in } G \text{ that are rooted in } v
\end{align}
We compute $c(\tc,v)$ by dynamic programming.
For each $v$ we initialize $c(\tc, v)=1$, where $T$ is the trivial treelet on $1$ node and $C = \{c_v\}$; all other counts are implicitly $0$.
Now suppose we have computed the counts of all treelets on $h-1$ nodes for some $h \le k$.
To compute $c(\tc, v)$ for some $\tc=(T,C)$ on $h$ nodes, we decompose $T$ in two smaller subtrees $\ta$ and $\tb$, rooted respectively at the root $r$ of $T$ and at a child of $r$, and combine their counts.
It is easy to see that $c(\tc, v)$ is given by (see~\cite{Bressan&2018b}):
\begin{equation}
\label{eqn:decomp}
c(\tc, v) = \frac{1}{\beta_T} \sum_{u \sim v} \sum_{\substack{C' \subset C \\ |C'|=|\ta|}}\!\!\!\! c(\tca, v) \cdot c(\tcb, u)
\end{equation}
where $\beta_T$ is the number of subtrees of $T$ isomorphic to $\tb$ rooted in a child of $r$.
In practice, one uses a \emph{canonical decomposition} which defines the pair $(\ta,\tb)$ uniquely as a function of $T$.
A simple analysis of the entire dynamic program shows that:
\begin{theorem}[\cite{Bressan&2018b}, Theorem 5.1]
The build-up phase of CC takes time $O(a^k m)$ and space $O(a^k n)$, for some constant $a>0$. 
\end{theorem}
In practice, the table size grows quickly.
For $k=6$ on a graph with 5M nodes, CC needs $50$GB of memory~\cite{Bressan&2018b}.

\subsection{The sampling phase}
\label{sub:sampling}
The goal of this phase is to estimate the graphlet counts by sampling graphlets from $G$.
We do this by sampling colorful treelet copies from the treelet count table, as follows.
First, draw a pair $(\tc, v)$ with probability proportional to $c(\tc,v)$.
This is possible since we know all the counts $c(\tc,v)$.
Now we want to sample a copy of $\tc = (T,C)$ rooted at $v$.
To this end, we take again the canonical decomposition of $T$ into $T'$ and $T''$.
We then sample a pair $(u, C'')$, where $u \sim v$ and $C'' \subset C$ contains $|T''|$ colors, with probability proportional to $\beta_T^{-1} c( (T', C \setminus C''), v) \cdot c( (T'', C'') , u)$. 
We then recursively sample a copy of $\tca = (T', C \setminus C'')$ rooted at $v$, and a copy of $\tcb = (T'', C'')$ rooted at $u$. Once we have the copies of $\tca$ and $\tcb$, we just combine them into a copy of $\tc$.
One can verify that the resulting copy is drawn uniformly at random from the set of all colorful treelets of $G$~\cite{Bressan&2018b}.

Using this $k$-treelet sampling primitive, one can estimate the copies of any given $k$-graphlet $H_i$ (e.g., the clique).
First, we estimate the number $c_i$ of \emph{colorful} copies of $H_i$.
To achieve this, we sample a treelet copy as described above.
This treelet copy necessarily spans some induced $k$-node subgraph $\glet$ of $G$.
Let $\chi_i$ be the indicator random variable of the event that $\glet$ is a copy of $H_i$.
It is easy to see that $\E[\chi_i] = c_i \sigma_i / t$, where $\sigma_i$ is the number of spanning trees in $H_i$ and  $t$ is the total number of colorful $k$-treelets of $G$.
Now, $t$ is known from the treelet count table, and $\sigma_i$ can be computed quickly via Kirchhoff's theorem (see below).
Therefore we can compute $\hat{c}_i = t\, \sigma_i^{-1} \chi_i$, which is an unbiased estimator of $c_i$.
By standard concentration bounds, $\hat{c}_i$ concentrates around $c_i$ if enough samples are taken.

Finally, to estimate the total number $g_i$ of copies of $H_i$, we simply divide $\hat{c}_i$ by the probability $p_k = k! / k^k$ that a fixed set of $k$ nodes in $G$ becomes colorful.
Indeed, if $G$ contains $g_i$ copies of $H_i$, then by linearity of expectation the expected number of copies of $H_i$ that become colorful is $\E[c_i] = p_k g_i$.
Therefore $\hat{g}_i = \hat{c}_i/p_k$ is an unbiased estimator for $g_i$.

\subsection{Statistical guarantees of the estimates}
The crucial point of the algorithm just described is the accuracy of the graphlet estimates, $\hat{g}_i$.
There are two sources of error: the coloring, which distorts the true graphlet distribution into the colorful one, and the sampling.

Regarding the coloring error, one can prove that the colorful graphlet distribution is statistically close to the actual graphlet distribution of $G$.
First, we report a bound from~\cite{Bressan&2018b}, slightly rephrased.
Let $g = \sum_i g_i$ be the total number of induced $k$-graphlet copies in $G$.
Then:
\begin{theorem}[\cite{Bressan&2018b} Thm 5.3]
\label{thm:color_conc}
For all $\epsilon>0$, a random coloring of $G$ with $k$ colors gives:
\begin{align}
    \prob\!\left[\left|\frac{c_i}{p_k} - g_i\right| > \frac{2 \epsilon g }{1-\epsilon}\right] = \exp\left(-\Omega\left( \epsilon^2 g^{1/k} \right)\right).
\end{align}
\end{theorem}
Note that the bound above is additive, that is, it establishes a deviation proportional to $g$, the total number of graphlets.
In this work we complement it with a multiplicative bound that holds when the maximum degree $\Delta$ of $G$ is small.
\begin{theorem}
\label{THM:CONC_DEP}
For all $\epsilon > 0$, a random coloring of $G$ with $k$ colors gives:
\begin{align}
\prob\!\left[\left|\frac{c_i}{p_k} - g_i \right| > \epsilon \,g_i \right] < 2 \exp\!\left(\!-\frac{2\epsilon^2 p_k^2 \, g_i }{(k-1)!\Delta^{k-2}}\right).
\end{align}
\end{theorem}
\begin{proof}
We use a concentration bound for dependent random variables from~\cite{Dubhashi&2009}.
Let $\vk_i$ be the set of copies of $H_i$ in $G$.
For any $h \in \vk_i$ let $X_{h}$  be the indicator random variable of the event that $h$ becomes colorful.
Let $c_i = \sum_{h \in \vk_i} X_h$; clearly $\E[c_i] = p_k g_i$.
Note that, for any $h_1,h_2 \in \vk_i$, the random variables $X_{h_1}, X_{h_2}$ are independent if and only if $|V(h_1) \cap V(h_2)| \le 1$, which means $h_1,h_2$ share at most one node.
For any $u,v \in G$ let then $g(u,v) = |\{h \in \vk_i : u,v \in h \}|$, and define $\chi_k = 1 + \max_{u,v \in G} g(u,v)$.
By a standard counting argument $\max_{u,v \in G} g(u,v) \le (k-1)!\Delta^{k-2}-1$ and thus $\chi_k \le (k-1)!\Delta^{k-2}$.
The bound then follows immediately from Theorem~3.2 of~\cite{Dubhashi&2009} by setting $t= \epsilon \E[c_i] = \epsilon p_k g_i$, $(b_{\alpha} - a_{\alpha}) = 1$ for all $\alpha = h \in \vk_i$, and $\chi^*(\Gamma) \le \chi_k \le (k-1)!\Delta^{k-2}$.
\end{proof}
These two bounds suggest that the random coloring does not introduce a significant distortion.
This is confirmed by our experiments, where the $\hat{g}_i$ appear concentrated around the mean.
Hence, one may avoid averaging over $\Theta(\exp(k))$ independent colorings, as suggested in the original color coding paper~\cite{Alon&1995}; one run is enough.
Thus, in a sense, the treelet count table is w.h.p.\ a database that holds (implicitly) a representative sample of all $k$-graphlets in $G$.

For what concerns the sampling error, standard concentration bounds apply to the error of uniform sampling (see above).
For AGS the analysis is more complex, and is shown in Section~\ref{sec:ags}.

\section{Fast construction of a compact treelet database}
\label{sec:speed}
This section details the internals of the build-up phase of \motivo\ and \motiveight.
Recall that the goal of this phase is to compute a compact database of colorful treelet counts (Section~\ref{sec:cc}).
This database will support the following operations:
\begin{itemize}
\item \texttt{occ($v$)}: get the total number of colorful treelet copies rooted at $v$
\item \texttt{sample($v$)}: get a uniform random colored treelet rooted at $v$
\end{itemize}
Using these two operations, we can implement graphlet sampling, by first selecting a node $v$ with probability proportional to the number of treelets rooted in it, and then sampling one such treelet uniformly at random.
For our adaptive sampling scheme AGS, we will also support the following operations:
\begin{itemize}
\item \texttt{occ($\tc, v$)}: get the total number of copies of $\tc$ rooted at $v$
\item \texttt{sample($T,v$)}: get a uniform random colored treelet rooted at $v$ with shape $T$
\end{itemize}

To give a detailed account of \motivo\ and \motiveight, we describe them incrementally. We start in Section~\ref{sub:treelets} by discussing the implementation of CC's build-up phase.
Then, in Section~\ref{sec:k16} (\motivo) and Section~\ref{sec:k8} (\motiveight) we progressively replace the algorithms and data structures of CC with ours, measuring the cumulative performance impact.
To measure the performance impact for \motivo, we use as baseline a \Cpp\ version of CC (which is originally in Java), that we wrote by carefully porting all algorithms and data structures.
The baseline for \motiveight\ is \motivo\ itself.
Finally, in Subsection~\ref{sub:build_opt} and Subsection~\ref{sub:biased} we describe additional optimizations for both \motivo\ and \motiveight.

\subsection{Treelets and counts}
\label{sub:treelets}
In this subsection we discuss the crucial aspects of the build-up phase of CC.
This phase spends nearly all its time in manipulating rooted (un)colored treelets and reading/writing their counts, as described in Section~\ref{sec:cc}.
Therefore, to speed up the phase, we need to make the manipulation of these objects as efficient as possible.

The first computationally intensive task is merging the treelet counts.
Consider a single treelet count $c(\tc,v)$.
To compute $c(\tc,v)$, both CC and our algorithms process all the neighbors $u$ of $v$ as follows.
First, suppose we have defined a total order over all the $\tc$ (this order is described below).
For every pair of non-zero counts $c(T'_{C'}, v)$ and $c(T''_{C''}, u)$, check that $C' \cap C'' = \emptyset$, and that $T''_{C''}$ is not smaller than the smallest subtree rooted in a child of the root of $T'_{C'}$.
Here, ``smaller'' and ``smallest'' are determined by the total order.
If these conditions hold, then $T'_{C'}$ and $T''_{C''}$ can be \emph{merged} into a treelet $\tc$.
It is easy to see that for each occurrence of $\tc$ rooted in $v$ there are exactly $\beta_T$ pairs of: (i) an occurrence of a treelet $T'_{C'}$ rooted in $v$, and (ii) an occurrence of a treelet $T''_{C''}$ rooted in a neighbor $u$ of $v$, for which the above procedure on $T'_{C'}$ and $T''_{C''}$ returns $\tc$. Consequently, the value of $c(\tc,v)$ is incremented by $\beta_T^{-1} c(T'_{C'},v) \cdot c(T''_{C''},u)$. This procedure is essentially a direct computation of the sum in Equation~\ref{eqn:decomp}, once a particular decomposition of $T_C$ has been chosen by our total order. Since Equation~\ref{eqn:decomp} holds regardless of the chosen decomposition, this implies the correctness of the build-up phase as well.
In this process, the computationally intensive part is the check-and-merge operation, which can be formalized as the primitive:
\begin{itemize}
\item \texttt{merge($T'_{C'}$,$T''_{C''}$)}: if possible, merge $T'_{C'}$, $T''_{C''}$ by appending $T''_{C''}$ to the root of $T'_{C'}$, else FAIL
\end{itemize}
In CC, this primitive is implemented as a recursive algorithm, which can be quite expensive.
This is because CC encodes each treelet $\tc$ as a classic, pointer-based tree data structure.
Here, we show a different implementation that makes the check-and-merge operation much faster.

The second computationally expensive task is storing and accessing the counts.
CC does it as follows: for each node $v \in G$, it keeps a dedicated hash table in main memory which maps each colored treelet $\tc > 0$ to its count $c(\tc,v)$.
In the hash table, the key used for $c(\tc,v)$ is a 64-bit pointer to an instance of $\tc$.
Thus, each entry of CC's table thus uses 128 bits: 64 for the key, and 64 for the integer count.
Already for $k=6$ on a graph with a few million nodes, storing the hash table can require a dozen GB of main memory, so this approach becomes quickly impractical.
Here, we show how to reduce the space used by the count table by almost an order of magnitude.

Before moving on, we note that a perfectly fair porting of CC is not possible.
This is because CC makes heavy use of fast specialized integer hash tables provided by the \texttt{fastutil}\footnote{\url{http://fastutil.di.unimi.it/}} library, which exists only in Java and seems to be crucial to its performance.
Indeed, for the porting we tested three popular libraries --- \texttt{google::sparse\_hash\_map} and \texttt{google::dense\_hash\_map} of the sparsehash library\footnote{\url{https://github.com/sparsehash/sparsehash}}, and \texttt{std::unordered\_map} from the \Cpp\ containers library.
With the first two, the porting is up to $17 \times$ slower than CC, and with the latter one it is up to $7\times$ slower.
Nonetheless, after all optimizations are in place, we are always faster than CC.

\subsection{\motivo: a general-purpose motif counter for \texorpdfstring{k $\le$ 16}{k <= 16} nodes}
\label{sec:k16}
We describe our first toolbox, \motivo, introduced in our preliminary work~\cite{Bressan&2019VLDB}, for graphlets on up to $16$ nodes.
Starting from the \Cpp\ porting of CC, we introduce succinct treelets and the succinct count table, obtaining the performance improvement shown in Figure~\ref{fig:impact1}.

\begin{figure}[ht]
\centering
\hfill
\includegraphics[scale=.7]{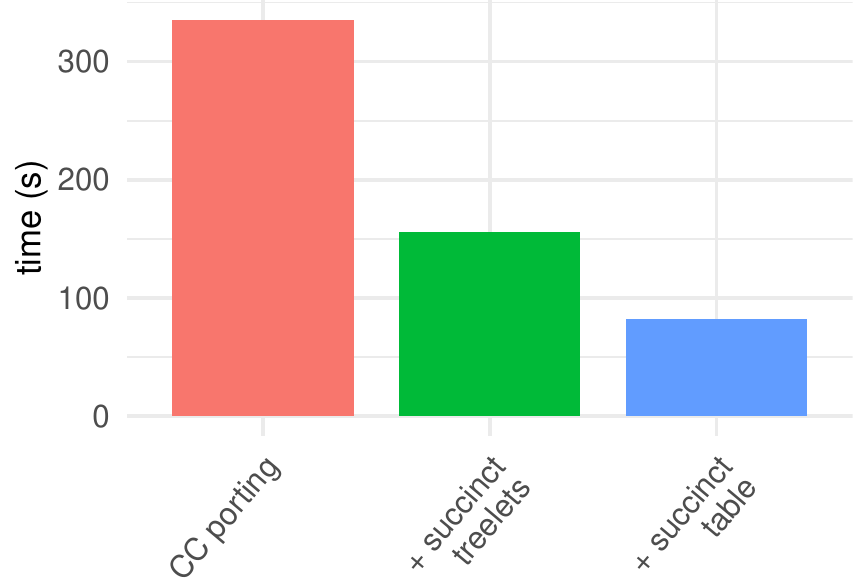}
\hfill
\includegraphics[scale=.7]{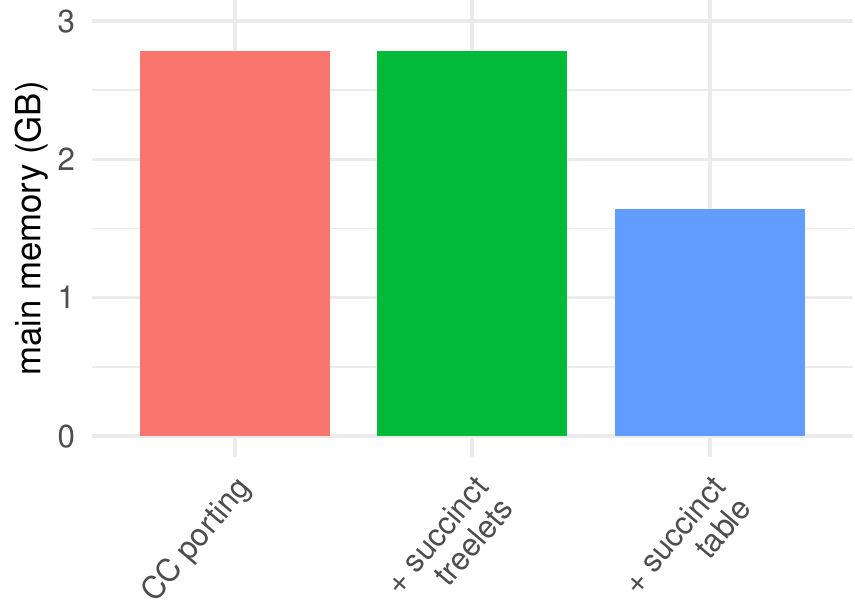}
\hfill
\caption{Cumulative impact of our optimizations on the build-up phase, for \amazon\ with $k=6$. The baseline is the \Cpp\ porting of CC.}
\label{fig:impact1}
\end{figure}

\subsubsection{Succinct treelets and count table}
\label{subsub:succinct}
\paragraph{Treelet representation.}
We drop CC's pointer-based structures, and we represent treelets as bitstrings.
This accelerates \texttt{merge($T$,$T'$)} by up to $150\times$ for $k=5$ and up to $1000\times$ for $k=7$.
Given an uncolored treelet $T$ rooted at $r$, its bitstring representation $s_{T}$ is defined as follows. 
Perform a DFS traversal of $T$ starting from $r$. 
The $i$-th bit of $s_{T}$ is $1$ (resp. $0$) if the $i$-th edge is traversed moving away from (resp. towards) $r$.

This encoding exhibits several nice properties.
First, for all $k \le 16$, it requires at most $30$ bits, which fit in a $4$-byte integer type.
Second, the total order over the treelets $T$ is just the lexicographic order over their encodings (and we can show that this coincides with the total order used by CC as well).
Third, the order serves also as tie-breaking rule for the DFS traversal: the children of a node are visited in the order given by their rooted subtrees.
This implies that every $T$ has a well-defined unique encoding $s_{T}$.
Fourth, merging $T'$ and $T''$ into $T$ boils down to concatenating the bitstrings $1,s_{T''},s_{T'}$, in this order, which is typically faster than manipulating a pointer-based structure.

To encode a \emph{colored} rooted treelet, $\tc=(T,C)$, we simply concatenate $s_{T}$ and the characteristic vector $s_C$ of $C$.\footnote{Given a universe $U={0,1,\dots, eta}$, the characteristic vector $\langle x_0, x_1, \dots, x_{\eta-1} \rangle$ of a subset $S \subseteq U$ contains one bit $x_i$ for each element in $U$, where $x_i$ is $1$ if $i \in S$ and $0$ otherwise.}
For all $k \le 16$, the resulting bitstring $s_{\tc}$ fits in 46 bits.
Set-theoretical operations on $C$ become bitwise operations over $s_C$ (\texttt{or} for union, \texttt{and} for intersection).
And again, the lexicographical order of the encodings gives the total order over the treelets. 
An example of a colored rooted treelet and its encoding is given in Figure~\ref{fig:treelet1}. 
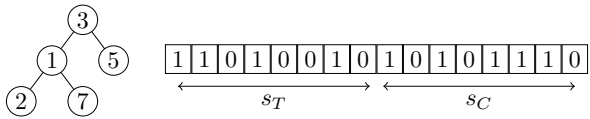
\begin{figure}[ht]
\centering
\scalebox{0.9}{

\begin{tikzpicture}[level distance=0.6cm,
sibling distance=1cm, minimum size=11pt,inner sep=1.5,outer sep=0,
level 1/.style={sibling distance=0.9cm, minimum size=12pt,outer sep=0},
level 2/.style={sibling distance=0.9cm, minimum size=12pt,outer sep=0}]

\tikzstyle{every node}=[circle,draw,minimum size=5pt]

\node (Root)  {3}
    child {
    node {1} 
    child { node {2} }
    child { node {7} }
}
child {
    node {5}
};
\end{tikzpicture}
\hspace*{5pt}
\begin{tikzpicture}
\matrix[matrix of nodes,row sep=0mm,set common column={1,...,16}{nodes={rectangle,draw,minimum width=1em,inner sep=2.8pt}}] (O)
{
1 & 1 & 0 & 1 & 0 & 0 & 1 & 0 &    1 & 0 & 1 & 0 & 1 & 1 & 1 & 0\\
};
\draw[<->](-2.9,-0.4) -- node[below] {$s_{T}$} (-0.1,-0.4);
\draw[<->](0.1,-0.4) -- node[below] {$s_C$} (2.9,-0.4);
\end{tikzpicture}}
\caption{A colored rooted treelet  $\tc$ and its encoding, shown for simplicity on just $8+8=16$ bits. The label $\ell$ of each node represents its color and the corresponding $\ell$-th least significant bit of $s_C$ is set to $1$. $s_T$ is the bitstring representing the tree $T$.}
\label{fig:treelet1}
\end{figure}

\begin{figure}[b]
    \centering
    \includegraphics[scale=0.5]{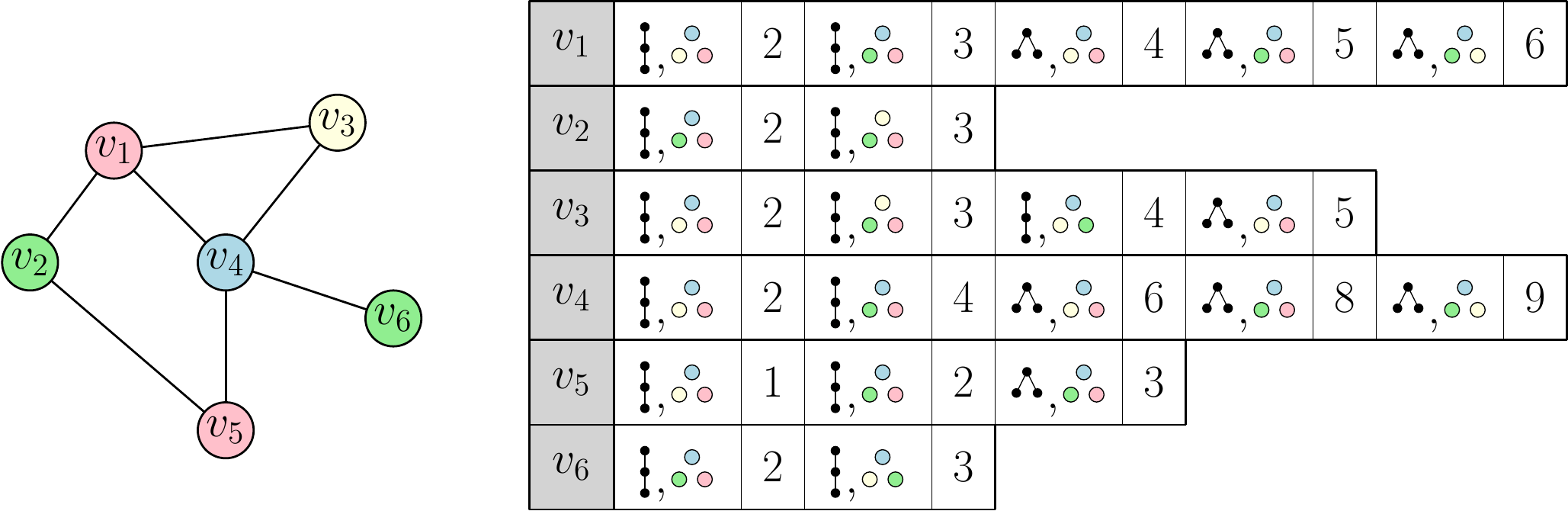}
\caption{Left: a graph $G$ whose vertices have been colored using $k=4$ colors. Right: a graphical representation of the count table (implicitly) storing the number $c(T_C, v)$ of all colored treelets $T_C$ of size $3$ in $G$. Notice how we actually store $\eta(T_C, v)$ instead of $c(T_C, v)$.}
    \label{fig:table_example}
\end{figure}

\paragraph{Count tables.}\hspace*{-7pt}
Instead of using hash tables, \motivo\ maintains the key-value pairs $(\tc, c(\tc, v))$ in a set of arrays, one for each $v \in G$ and each treelet size $h \in \{1,\ldots,k\}$, storing only the entries with a non-zero count $c(\tc, v)> 0$ (note that each array is built only once and never modified thereafter).
Clearly, using arrays makes iteration extremely fast.
Moreover, given the key of an entry, we do not need to dereference any pointer to access the corresponding colored treelet ---the key is the treelet itself.
The price to pay is that searching for a given key is potentially more expensive.
However, by sorting every array according to the total order over the keys (see above), we can find a key via binary search in time $O(k)$, since every array has length $O(6^k)$.\footnote{By Cayley's formula: there are $O(3^k k^{-3/2})$ rooted treelets on $k$ vertices~\cite{otter1948number}, and $2^k$ subsets of $k$ colors.}

For what concerns the treelet counts themselves, while CC uses 64 bits per count, \motivo\ uses 128 bits per count.
This adds a small overhead,\footnote{Tests on our machine show that summing 500k unsigned integers is 1.5$\times$ slower with 128-bit than with 64-bit integers.} but avoids dangerous overflows for large values of $k$.\footnote{Already for $k=6$, the number of $k$-stars centered in a node of degree $2^{16}$ is $\approx\!10^{22}$.}
Moreover, since each key fits into 48 bits, each array entry actually uses 48+128=176 bits.
We conclude with one final optimization.
In place of $c(\tc, v)$, \motivo\  stores the \emph{cumulative} count $\eta(\tc,v) = \sum_{{T'}_{\!C'} \le \tc} c(T'_{C'}, v)$, see Figure~\ref{fig:table_example} for a toy example.
This has several benefits: (i) the count $c(\tc, v)$ can be computed with negligible overhead as the difference between two consecutive entries, (ii) the total treelet count $\eta_v$ for $v$ is at the end of the record, and (iii) we can sample a uniform random treelet rooted in $v$ in time $O(k)$, by drawing a uniform random value $X$ in $\{1,\ldots, \eta_v\}$ and then doing a binary search for $X$ over the array of $v$.

\subsection{\motiveight: counting 8-graphlets in large graphs}
\label{sec:k8}
In this subsection we describe \motiveight, a version of \motivo\ specialized for $k \le 8$.
To this end, we start with \motivo\ and plug in three new ingredients which improve its performance ---see Figure~\ref{fig:impact2}.

\subsubsection{Integer treelet encoding (ITE)}
\label{subsub:ite}
Our first ingredient is an extremely compact representation of treelets as unsigned integers.
The idea is simple: we observe that there are exactly $1991$ rooted colored treelets on at most $8$ nodes.
Therefore, we can encode each treelet $T_{C}$ as an $11$-bit integer.
This reduces the space by more than $4\times$ compared to the bitstring encoding of \motivo.
From an asymptotic point of view, we are using the optimal amount of bits, up to a multiplicative factor $1+o(1)$.
In this sense, one cannot use fewer bits per treelet.
To ensure treelets are memory-aligned, we pad the $11$-bit representation into a $16$-bit integer.
This leaves $5$ spare bits, that we use to encode the length of our variable-length count, see described in Section~\ref{sec:vle}.
The impact of ITE on the overall space usage of \motivo\ is around $\simeq 20\%$ on all instances (see Figure~\ref{fig:impact2}).

Adopting ITE makes treelet manipulations harder.
Recall for example that \texttt{merge($T'_{C'}, T''_{C''}$)} checks if $T'_{C'}$ and $T''_{C''}$ can be merged into a treelet $\tc$, and that $T$ can be decomposed into $T'$ and $T''$.
Since in ITE a treelet is just a number, \texttt{merge} would need to go back and forth between ITE and bitstring encodings.
To resolve, we \emph{precompute} the results of \texttt{merge($T'_{C'}, T''_{C''}$)} on all treelets on up to $8$ nodes, and store all the results, in ITE format, in a bidimensional array.
In this array, entry $[i][j]$ is the ITE index of the treelet resulting from merging $i$ and $j$ (or $-1$ for FAIL).
Another array tells, for every ITE index of a treelet $T$, the ITE indices $i,j$ of the treelets of the canonical decomposition of $T$.
Other arrays tell us the ITE index of the treelet $T$ and the set of colors $C$ associated to a colored treelet $T_C = (T,C)$, and so on.
Thus, each treelet operation is just a sequence of black-blox lookups in these arrays, using only the ITE encoding.
The total size of all arrays is less than $3$ megabytes, which fits in the cache of a modern CPU.
With this technique, the time taken by a single treelet operation is similar to the one required by the bitstring encoding of Section~\ref{subsub:succinct}.

\begin{figure}[ht]
\hfill
\includegraphics[height=100pt]{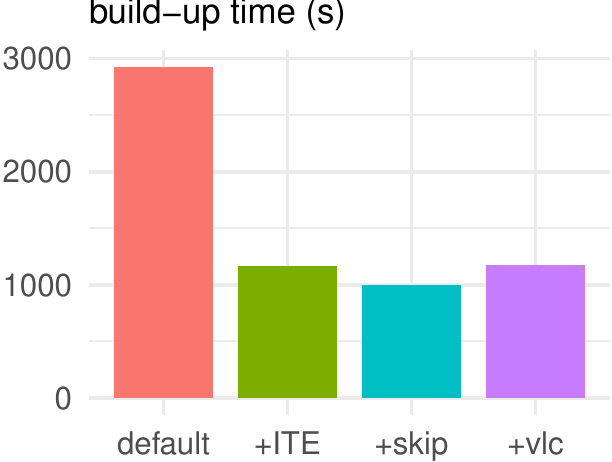}
\hfill
\includegraphics[height=100pt]{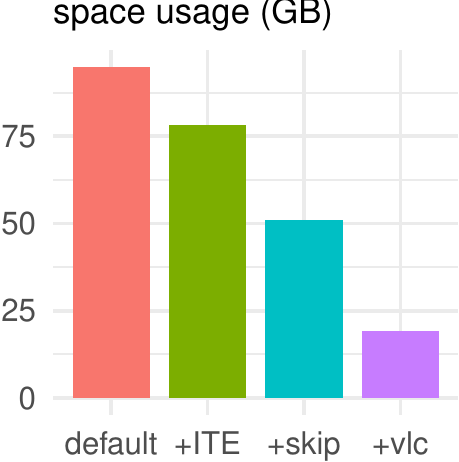}
\hfill
\includegraphics[height=100pt]{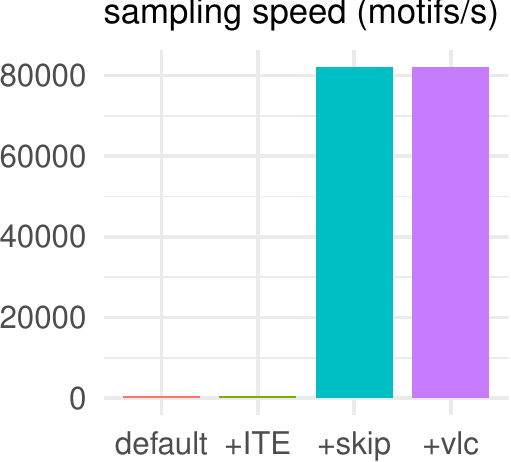}
\hfill\phantom{.}
\caption{\twitter\ graph, $k=6$. Cumulative impact of using ITE, skipping the heaviest round of the dynamic program, and using variable-length counts. The ``default'' baseline is \motivo, described in the previous section.}
\label{fig:impact2}
\end{figure}

\subsubsection{Round skipping via balanced treelet decomposition}
\label{subsub:balanced}
Recall again that, in the build-up phase, we repeatedly merge the counts of smaller treelets.
In particular, in round $k-1$ we produce the counts for all treelets on $k-1$ nodes.
For $k \le 8$, round $k-1$ is consistently the most expensive one, and consumes roughly $40\%$ of the time and space of the entire phase --- see Figure~\ref{fig:rounds}.
We would like to reduce the time and space usage of this round.

Our observation is the following: \emph{a $(k-1)$-treelet count is used only to compute the counts of $k$-treelets whose smallest subtree in the canonical decomposition is a single node}.
For example, a $k$-star is certainly decomposed into a single node and a treelet on $k-1$ nodes, and therefore to count $k$-stars we need the $(k-1)$-treelet counts.
It turns out that stars are the only treelets whose decomposition necessarily contains a treelet on a single node; for all other treelets, we can find a decomposition where the largest subtree contains at most $k-2$ nodes.
And for stars, we can avoid counting them at all: at sampling time, they can be sampled very efficiently in the naive way.

Let us now describe our round skipping technique in more detail.
We rely on the following basic fact:
\begin{lemma}
Any $k$-treelet $T$ that is not a star has an edge that, if cut, yields two trees each with at most $k-2$ nodes.
\end{lemma}
\begin{proof}
If $T$ is not a star then it contains a path on $3$ edges, say $(x,u,v,y)$.
Cutting $(u,v)$ yields two trees, each one with at least $2$ nodes or, equivalently, at most $(k-2)$ nodes.
\end{proof}
Therefore, for any non-star $k$-treelet $T$, we can find a root $u$ and a child $v$ of $u$ that give a \emph{balanced decomposition} of $T$ into $T'$ and $T''$, that is, one where both $T'$ and $T''$ have at most $\le k-2$ nodes, see Figure~\ref{fig:balanced}.
We therefore replace the canonical decomposition of Subsection~\ref{subsub:succinct} with this balanced decomposition, and we exclude $k$-stars from the set of treelets that we count in the $k$-th round.
Then, we completely skip round $(k-1)$ of the dynamic program.
This yields a reduction in both space and time, consistently across all instances, of up to $40\%$ (Figure~\ref{fig:rounds}).

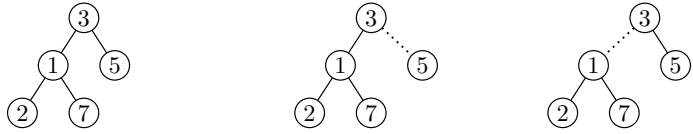
\begin{figure}[h]
\centering\scalebox{.9}{\begin{tikzpicture}[line width=.5pt,level distance=0.7cm,
sibling distance=1cm, minimum size=10pt,inner sep=1.5,outer sep=0,
level 1/.style={sibling distance=.9cm, minimum size=11pt,outer sep=0},
level 2/.style={sibling distance=.9cm, minimum size=11pt,outer sep=0},
every node/.style={solid,black},
norm/.style={solid,black,line width=.5pt},
cut/.style={dotted,thick}
]

\tikzstyle{every node}=[circle,draw,minimum size=5pt]

\node (r1) at(0,0) {3}
    child {
    node {1} 
    child { node {2} }
    child { node {7} }
}
child {
    node {5}
};

\node (r2) at($(r1)+(4.2,0)$) {3}
    child {
    node {1} 
    child { node {2} }
    child { node {7} }
} 
child[cut,sibling distance=1.5cm] {
    node[norm] {5}
};

\node (r3) at($(r2)+(4,0)$) {3}
    child[cut,sibling distance=1.5cm] {
    node[norm] {1} 
    child[norm] { node {2} }
    child[norm] { node {7} }
}
child {
    node {5}
};

%
%
%
%
%
%

\end{tikzpicture}}
\caption{A rooted colored $5$-treelet (left) with its original decomposition in two subtrees on $4$ and $1$ nodes (middle) and its balanced decomposition in two subtrees on $3$ and $2$ nodes (right).}
\label{fig:balanced}
\end{figure}

\begin{figure}[ht]
\centering
\includegraphics[scale=.7]{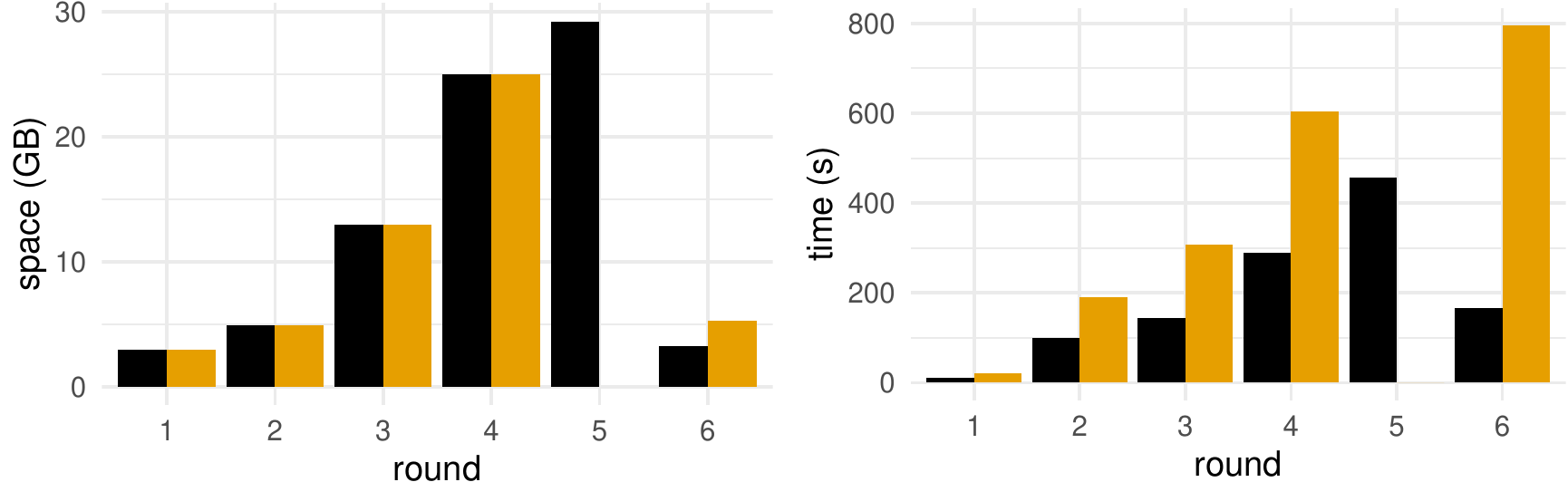}
\caption{Build-up phase on \twitter\ for $k=6$: space and time usage of single rounds, before and after adopting balanced treelet decompositions and round skipping.}
\label{fig:rounds}
\end{figure}

As said, since we do not have the counts of $k$-stars anymore, to sample a star we now simply draw a root node $v$ from $G$ with probability proportional to $\binom{d_v}{k-1}$ and then choose $k-1$ neighbors of $v$ u.a.r.\ without replacement.
This gives an \emph{uncolored} $k$-star u.a.r.\ from $G$, which is even better since we avoid the noise introduced by coloring.
Sampling in this fashion is much faster than using the count table, and since most treelets of $G$ are often stars, the sampling rate can increase by several orders of magnitude (see Figure~\ref{fig:impact2}).

\subsubsection{Variable-length counts}
\label{sec:vle}
Our third and final ingredient is aimed again at saving space. 
To this end, we encode each treelet count $c(\tc,v)$ as a \emph{variable-length count}.
The rationale is that, in practice, most of the treelet counts can be stored in very few bits as shown, e.g., by Figure~\ref{fig:rl}.
Our variable-length counts can be thought of as a variant of Elias delta coding \cite{Elias}, which can represent any integer $x$ using only $(1+o(1))\lceil \log_2 x\rceil$ bits.
In practice, by using the $5$ spare bits left by the encoding of $\tc$ (see Section~\ref{subsub:ite}), we encode each key-value pair $(\tc, c(\tc, v))$ in a byte-aligned memory region, as shown in Figure~\ref{fig:vlc}, that is we use:
\begin{itemize}
\item $11$ bits for the integer representation of $\tc$, see Section~\ref{subsub:ite}.
\item $5$ bits for the length $\ell$ of $c(\tc,v)$, expressed in bytes.
Thus, we can support counts $c(\tc,v)$ on up to $256$ bits (twice the maximum count length supported by \motivo), and with values as large as $2^{256}-1$. 
\item $\ell$ bytes for the binary encoding of $c(\tc,v)$
\end{itemize}

Variable-length counts yield an additional space reduction of $\ge 60\%$ on all graphs for all $k \ge 6$.
The downside is that we cannot find counts via binary search, as the counts are not aligned in memory anymore.
Moreover, we pay the obvious overhead of encoding and decoding the counts.
This has a significant impact on the build-up time; in the worst case, we witness an increase of about $50\%$.
However, the gains outweigh the losses: on our largest graphs \twitter\ and \fs, variable-length counts are crucial to reduce the space footprint enough to manage graphlets on $k=8$ nodes.

\begin{figure}[h!]
\centering\includegraphics[height=100pt]{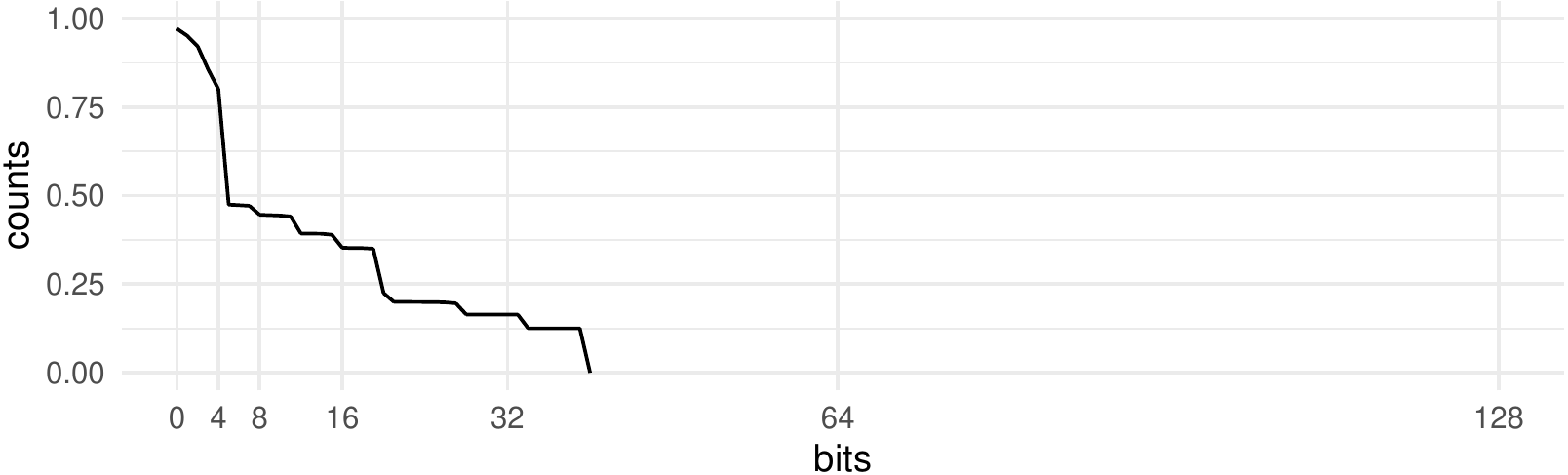}
\caption{The fraction of treelet counts that require at least $b$ bits, as a function of $b$, for the treelet table of \twitter\ with $k=6$. The resulting average bit length is just $14$, almost $90\%$ less than $128$ bits.
}
\label{fig:rl}
\end{figure}

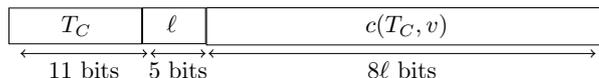
\begin{figure}[h!]
\centering\scalebox{.9}{\begin{tikzpicture}
\matrix[matrix of nodes,row sep=0mm,set common column={1,2,3}{nodes={rectangle,draw,minimum width=1em,inner sep=2.8pt,minimum height=15pt}}] (O)
{
\hspace*{19pt}$\tc$\hspace*{18pt} & \;\;\,$\ell$\phantom{$\tc$}\!\! & \hspace*{60pt} $c(\tc,v)$ \hspace*{60pt}\\
};
\draw[<->](-4.2,-0.4) -- node[below] {$11$ bits} (-2.35,-0.4);
\draw[<->](-2.3,-0.4) -- node[below] {$5$ bits} (-1.5,-0.4);
\draw[<->](-1.45,-0.4) -- node[below] {$8 \ell$ bits} (4.2,-0.4);
\end{tikzpicture}}
\caption{Variable-length encoding of a treelet count.}
\label{fig:vlc}
\end{figure}

\subsection{Lower-level optimizations and architectural details}
\label{sub:build_opt}
For completeness and reproducibility, we describe some additional optimizations and features of \motivo\ and \motiveight, some with a significant impact.

\subsubsection{Zero-rooting}
This is only for \motivo.
Consider a colorful treelet copy in $G$ formed by the nodes $v_1,\ldots,v_h$.
This treelet appears in the records of $v_1,\ldots,v_h$, since it counts as a rooted treelet for each of them.
Therefore, the treelet is counted $h$ times.
This redundancy is necessary when $h < k$, since we need all rootings for the next round of the dynamic program, see~\eqref{eqn:decomp}.
However, for $h=k$ this is useless.
Thus, we store $k$-treelet counts only at nodes of color $0$.
This cuts the running time by $30\%-40\%$, while reducing the size of the $k$-treelets records by a factor of $k$, and the total space usage by $\approx 10\%$.
Notice that \motiveight\ already counts $k$-treelets only once, thanks to the balanced treelet decomposition of Section~\ref{subsub:balanced}.

\subsubsection{Greedy flushing}
To reduce the memory footprint, we use a greedy flushing strategy.
Suppose we are building the count table of the $h$-treelets.
We temporarily store the record/array of $v$ in a hash table, which allows for efficient insertions and lookups; when done, we immediately flush it on disk and delete the hash table.
In this way we produce the count tables for all nodes of $G$, in some order.
Thus, a second I/O pass is needed to sort the tables by their corresponding node of $G$, so that they can be retrieved efficiently in the next round.
This technique increased the total runtime by at most 10\% in all our runs.

\subsubsection{Multi-threading}
We make heavy use of thread-level parallelism in both the build-up and sampling phases.
For the build-up phase, the count of each node $v \in G$ is computed by an independent thread from a thread pool of a fixed size.
As long as the number of remaining vertices is sufficiently large,
each thread is assigned a (yet unprocessed) vertex $v$ and will compute all the counts $c(\tc, v)$ for all pairs $\tc$.
Obviously, when the number of remaining vertices drops below the number of available threads, some threads become idle.
When this happens, we partition the edges of a single vertex $v$ across different threads and make them compute different summands of the outermost sum of Equation~\eqref{eqn:decomp}.
The partial sums are then summed together into $c(\cdot, v)$.
For the sampling phase, samples are by definition independent and are taken by all threads in parallel.

\subsubsection{Memory-mapped reads}
Recall that our treelet count database is stored in external memory.
This entails I/O access, since
computing the count table for treelets of size $h$ requires the count tables of each size $j < h$.
We delegate the task to the operating system by using memory-mapped I/O.
This means that we see the count tables as if they resided in main memory, and the operating system takes care of loading and storing them to disk.
With enough memory this gives virtually no overhead; otherwise, the OS will reclaim memory by unloading part of the tables, and future requests to those parts will incur a page fault and prompt a reload from the disk.
The overhead in terms of additional I/O turns out to be at most 100MB, except for $k=8$ on \lj\ (34GB) and \yelp\ (8GB) and for $k=6$ on \fs\ (15GB).
In these cases the overhead is inevitable, as the aggregate size of the tables is close to or even larger than the memory size.

\subsection{Biased coloring}
\label{sub:biased}
Finally, we describe a simple trick that reduces space significantly, in exchange for accuracy, which is useful on very large graphs.
The idea is to skew the distribution of colors so that fewer treelets become colorful and we have less counts to process and store.

Consider the following color distribution.
We choose each color $i \in \{1,\ldots,k-1\}$ with probability $\lambda \ll \frac{1}{k}$, and color $0$ with probability $1-\lambda(k-1)$.
Then, for any set of $j$ colors $C$, the probability that a given $j$-treelet is colored with $C$ is:
\begin{align}
p_{k,j}(C) = \left\{
\begin{array}{ll}
j! \lambda^j &\; \text{if} \; 0 \notin C \\
\simeq j! \lambda^{j-1} &\; \text{if} \; 0 \in C
\end{array}
\right.
\end{align}
Therefore, if $\lambda$ is sufficiently small, for most treelets we will have a zero count at $v$.
Moreover, most nonzero counts will be for a restricted set of colorings -- those containing color $0$.
This reduces the size of the treelet count table, and thus the running time of the algorithm.
As a downside we suffer a loss of accuracy, since a lower colorful probability increases the variance of the number $c_i$ of colorful copies of $H_i$.
However, if $n$ is large enough, and a substantial fraction of nodes of $G$ is part of some occurrence of $H_i$, then the \emph{total} number of copies $g_i$ of $H_i$ is large enough to ensure concentration.
In particular, by Theorem~\ref{THM:CONC_DEP} the accuracy loss is negligible as long as $\lambda^{k-1} n / \Delta^{k-2}$ is large.

This technique allows us to manage our largest instances, \twitter\ and \fs\ for $k=8$, with an acceptable loss of accuracy, see Section~\ref{sec:exp}.
In those experiments, we use an educated guess of $\lambda=0.001$.
We note that, otherwise, one could find a good value for $\lambda$ by setting $\lambda \ll 1/kn$ and then growing $\lambda$ until a good fraction of counts are positive, at which point Theorem~\ref{THM:CONC_DEP} ensures concentration.

\section{Sampling treelets from the database}
\label{sec:sample}
This section describes in detail the algorithms for sampling graphlets from the treelet count table.
Recall that we support two sampling algorithms: uniform sampling, which is the native sampling algorithm of CC, and our novel adaptive graphlet sampling strategy (AGS).
Uniform sampling is exactly the one described in Section~\ref{sec:cc}.
Hence we directly move on to AGS in the next section; we then conclude by describing lower-level optimizations that apply to both sampling strategies, and, in many cases, increment the sampling rate substantially.

\subsection{Adaptive Graphlet Sampling (AGS)}
\label{sec:ags}
Recall that the main idea of CC is to build a compact database for sampling $k$-treelets from $G$.
Interestingly, we can choose the kind of treelet to sample (a star, a path, etc.).
That is, for every $k$-treelet $T$ our database supports the operation:
\begin{itemize}[itemsep=0pt,parsep=0pt]
\item \sample($T$): return a colorful copy of $T$ u.a.r.\ from $G$
\end{itemize}
We can use this primitive to virtually ``delete'' certain graphlets from the database, and focus on other ones.

Let us explain the idea with an example.
Suppose $G$ contains just two types of colorful graphlets, $H_1$ and $H_2$, of which $H_2$ represents a tiny fraction, say $10^{-10}$.
With uniform sampling, we will need approximately $10^{10}$ samples before finding $H_2$.
Suppose, however, that $H_1$ and $H_2$ are spanned by treelets of different shape, say $T_1$ and $T_2$.
We can then start calling \sample($T_1$), which will return only copies of $H_1$, until we estimate accurately $H_1$.
At this point we call \sample($T_2$), which will return only copies of $H_2$, until we estimate accurately $H_2$, too.
Thus, using \sample($T$) we can estimate both graphlets with just $O(1)$ samples.
Clearly, the general situation is more complex, as we have thousands of graphlets with common spanning treelets.
Still, the idea can be adapted and it works strikingly well.

Let us describe AGS in more detail.
We start by invoking \sample($T$) on the most frequent $k$-treelet $T$ in $G$ (which we know from the database).
Eventually, some graphlet $H_i$ spanned by $T$ will appear enough times, say $\Theta(\epsilon^{-2}\ln(\nicefrac{1}{\delta}))$, so that we can estimate its occurrences with a multiplicative approximation of $1+\epsilon$, with a probability of at least $1-\delta$.
We then say $H_i$ is \emph{covered}.
Now we do not need any additional sample of $H_i$, so we would like to ``delete'' it.
That is, we want to switch to another treelet $T'$ that does \emph{not} span $H_i$.
Such a $T'$ may not exist, but we can use the $T'$ that \emph{minimizes} the probability of returning a copy of $H_i$.
Now the crucial point is that we can find $T'$ as follows.
First, we have a good estimate $\hat{c}_i$ of the number of colorful copies of $H_i$.
Then, for each $k$-treelet $T_j$ we can estimate the number of colorful copies of $T_j$ that span a colorful copy of $H_i$ in $G$ as $\hat{c}_i\, \sigma_{ij}$, where $\sigma_{ij}$ is the number of spanning trees of $H_i$ isomorphic to $T_j$.
Finally, dividing this estimate by the number $t_j$ of colorful copies of $T_j$ in $G$ yields an estimate of the probability that \sample($T_j$) spans a copy of $H_i$.
That is,
\begin{align}
\Pr(\sample(T_j) \text{ yields a copy of } H_i) = \frac{\text{\# of colorful copies of }T_j\text{ in G spanning }H_i}{\text{\# of colorful copies of }T_j\text{ in }G} = \frac{\hat{c}_i \sigma_{ij}}{t_j}
\end{align}
More generally, we need the probability that \sample($T_j$) spans a copy of some covered graphlet:
\begin{align}
\Pr(\sample(T_j) \text{ yields a covered graphlet}) = \frac{1}{t_j}\sum_{H_i \text{ covered}}\hat{c}_i \sigma_{ij}
\end{align}
We switch to the treelet $T_{j^*}$ minimizing this probability, and continue sampling until a new graphlet becomes covered.

The pseudocode of AGS is listed below. 
A graphlet is marked as covered when it has appeared in at least $\bar{c}$ samples.
To have a probability of at least $1-\delta$ of obtaining a multiplicative $(1+\epsilon)$-approximation over all $k$-graphlets one would set $\cerr = O(\epsilon^{-2}\ln(\nicefrac{s}{\delta}))$ where $s=s_k$ is the number of distinct $k$-graphlets.
In our experiments we set $\bar c = 1000$, which gives good accuracy on most graphlets.
We denote by $H_1,\ldots,H_s$ the distinct $k$-node graphlets and by $T_1,\ldots,T_{\varsigma}$ the distinct $k$-node treelets.

\renewcommand{\thealgorithm}{}
\begin{algorithm}[h!]
\caption{AGS($\epsilon, \delta$)}
\begin{algorithmic}[1]
\small
\State $(c_1,\ldots,c_s) \leftarrow (0,\ldots,0)$ \Comment{graphlet counts}
\State $(w_1,\ldots,w_s) \leftarrow (0,\ldots,0)$ \Comment{graphlet weights}
\State $\cerr \leftarrow \lceil\frac{4}{\epsilon^2}\ln(\frac{2s}{\delta})\rceil$ \Comment{covering threshold}
\State $C \leftarrow \emptyset$ \Comment{graphlets covered}
\State{$T_j \leftarrow$ an arbitrary treelet type}
\While{$|C| < s$}
\For{each $i'$ in $1,\ldots,s$}
\State $w_{i'} \leftarrow w_{i'} + \sigma_{i'j}/t_j$ \label{ags:w_j}
\EndFor
\State $T_G \leftarrow$ an occurrence of $T_j$ drawn u.a.r.\ in $G$
\State $H_i \leftarrow$ the graphlet type spanned by $T_G$ \label{ags:h_j}
\State $c_i \leftarrow c_i + 1$ \label{ags:c_j}
\If{$c_i \ge \cerr$} \Comment{switch to a new treelet $T_j$}
  \State $C \leftarrow C \cup \{ i \}$
  \State $j^* \leftarrow \arg \min_{j'=1,\ldots,{\varsigma}} \frac{1}{t_{j'}} \sum_{i' \in C} \sigma_{i'j'} \, c_{i'} / w_{i'}$ \label{ags:estim}
  \State $T_j \leftarrow T_{j^*}$ 
\EndIf
\EndWhile
\State \textbf{return} $(\frac{c_1}{w_1}, \ldots, \frac{c_s}{w_s})$
\end{algorithmic}
\end{algorithm}

\subsection{Approximation guarantees of AGS}
This section is dedicated to showing that AGS provides strong statistical guarantees.
Our main result is that, if AGS chooses the ``right'' treelet $T_{j^*}$, then we obtain multiplicative error guarantees for all graphlets at once.
Formally:
\begin{theorem}
\label{THM:AGS_APX}
If the tree $T_{j^*}$ chosen by AGS at line~\ref{ags:estim} minimizes $\prob[$\sample$(T_j)$ spans a copy of some $H_i \in C]$ then, with probability $(1-\delta)$, when AGS stops, $(c_i/w_i)$ is a multiplicative $(1\pm\epsilon)$-approximation of $g_i$ for all $i=1,\ldots,s$.
\end{theorem}
The proof requires a martingale analysis, since the distribution from which we draw the graphlets changes over time.
To this end, from now on we fix a graphlet $H_i$ and analyse the concentration of its estimate.
We drop the index $i$ from the notation unless necessary.
We start by recalling the following martingale tail inequality from~\cite{Alon&2010}:
\begin{theorem}[\cite{Alon&2010}, Theorem 2.2]
\label{thm:alon}
Let $(Z_0,Z_1,\ldots)$ be a mar\-tin\-gale with respect to the filter $(\mathcal{F}_\tau)_{t \ge 0}$. Suppose that $Z_{\tau+1} - Z_\tau \le M$ for all $\tau$, and write $V_t = \sum_{\tau=1}^t \var{Z_\tau|\mathcal{F}_{\tau-1}}$. Then for any $z,v>0$ we have:
\begin{align}
\prob\left[\exists \, t  : Z_t \ge Z_0 + z, V_t \le v\right] \le \exp{\left(-\frac{z^2}{2(v+Mz)}\right)}
\end{align}
\end{theorem}
We now plug the appropriate quantities from our algorithm into Theorem~\ref{thm:alon}.
\begin{enumerate}[label=\Alph*),itemsep=2pt,leftmargin=15pt]
    \item For $t \ge 1$, let $X_t$ be the indicator random variable of the event ``$H_i$ is the graphlet sampled at step $t$'' (line~\ref{ags:h_j} of AGS)
    \item For $t \ge 0$, let $Y_j^t$ be the indicator random variable of the event ``at the end of step $t$, the treelet to be sampled at the next step is $T_j$''
    \item For $t \ge 0$ let $\mathcal{F}_t$ be the event space generated by the random variables $Y_j^{\tau} : j\in[\varsigma], \, \tau=0,\ldots,t$
    \item For any random variable $Z$, then, $\E[Z \,| \,\mathcal{F}_t] = \E[Z \, | \, Y_j^{\tau} : j\in[\varsigma], \, \tau =0,\ldots,t]$, and $\var{Z\,|\,\mathcal{F}_t}$ is defined analogously
    \item For $t \ge 1$ let $P_t = \E[X_t|\mathcal{F}_{t-1}]$ be the probability that the graphlet sampled at the $t$-th invocation of line~\ref{ags:h_j} is $H_i$, as a function of the events up to time $t-1$.
It is immediate to see that $P_t = \sum_{j=1}^{\varsigma} Y_j^{t-1} a_{ji}$
    \item Let $Z_0 = 0$, and for $t \ge 1$ let $Z_t = \sum_{\tau=1}^t (X_t - P_t)$.
Now, $(Z_t)_{t \ge 0}$ is a martingale with respect to the filter $(\mathcal{F}_t)_{t\ge 0}$, since $Z_t$ is obtained from $Z_{t-1}$ by adding $X_t$ and subtracting $P_t$ which is precisely the expectation of $X_t$ w.r.t.\ $\mathcal{F}_{t-1}$
    \item Let $M=1$, since $|Z_{t+1}-Z_t| = |X_{t+1} - P_t| \le 1$ for all $t$
\end{enumerate}
Finally, notice that $\var{Z_t|\mathcal{F}_{t-1}} = \var{X_t|\mathcal{F}_{t-1}}$, since again $Z_t = Z_{t-1} + X_t - P_t$, and both $Z_{t-1}$ and $P_t$ are constant over $\mathcal{F}_{t-1}$, so their variance w.r.t.\ $\mathcal{F}_{t-1}$ is $0$.
Now, $\var{X_t|\mathcal{F}_{t-1}} = P_t(1-P_t) \le P_t$; and therefore we have $V_t = \sum_{\tau=1}^t \var{Z_\tau \,|\,\mathcal{F}_{\tau-1}} \le \sum_{\tau=1}^t P_{\tau}$.
By applying Theorem~\ref{thm:alon} above, we obtain:
\begin{align}
\label{eq:martbound}
\prob\left[\exists \, t : Z_t \ge z, \sum_{\tau=1}^t P_\tau \le v \right] &\le 
\exp{\left(-\frac{z^2}{2(v+z)}\right)} && \forall \; z,v > 0
\end{align}
Now consider AGS($\epsilon, \delta$).
Recall that we are looking at a \emph{fixed} graphlet $H_i$ (which here does \emph{not} denote the graphlet sampled at line~\ref{ags:h_j}).
Note that $\sum_{\tau=1}^t X_{\tau}$ is exactly the value of $c_i$ after $t$ executions of the main cycle (see line~\ref{ags:c_j}).
Similarly, $\sum_{\tau=1}^t P_\tau$ is the value of $g_i \cdot w_i$ after $t$ executions of the main cycle: indeed, if $Y_j^{t-1}=1$, then at step $\tau$ we add to $w_i$ the value $\frac{\sigma_{ij}}{t_j}$ (line~\ref{ags:w_j}), while the probability that a sample of $T_j$ yields $H_i$ is exactly $\frac{g_i\sigma_{ij}}{t_j}$.
Therefore, after the main cycle has been executed $t$ times, $Z_t = \sum_{\tau=1}^t (X_t - P_t)$ is the value of $c_i - g_i w_i$.

We now derive our concentration bounds.
Suppose that, when AGS($\epsilon, \delta$) returns, $\frac{c_i}{w_i} \ge g_i(1 + \epsilon)$, i.e., $c_i(1 - \frac{\epsilon}{1+\epsilon}) \ge g_i w_i$.
On the one hand this implies that $c_i - g_iw_i \ge c_i\frac{\epsilon}{1+\epsilon}$, i.e., $Z_t \ge c_i\frac{\epsilon}{1+\epsilon}$; and since upon termination $c_i = \cerr$, this means $Z_t \ge \cerr\frac{\epsilon}{1+\epsilon}$.
On the other hand it implies that $g_i w_i \le c_i(1 - \frac{\epsilon}{1+\epsilon})$, i.e., $\sum_{\tau=1}^t P_\tau \le c_i(1 - \frac{\epsilon}{1+\epsilon})$; again since upon termination $c_i = \cerr$, this means $\sum_{\tau=1}^t P_\tau \le \cerr(1 - \frac{\epsilon}{1+\epsilon})$.
We can then apply \eqref{eq:martbound} with $z=\cerr\frac{\epsilon}{1+\epsilon}$ and $v=\cerr(1 - \frac{\epsilon}{1+\epsilon})$, and since $v+z=\cerr$ we get:
\begin{align}
\prob\!\left[\frac{c_i}{w_i} \ge g_i(1+\epsilon)\right]
&\le \exp{\!\left(-\frac{(\cerr\frac{\epsilon}{1+\epsilon})^2}{2\cerr}\right)}
= \exp{\!\left(-\frac{\epsilon^2 \cerr}{2(1+\epsilon)^2}\right)}
\end{align}
but $\frac{\epsilon^2 \cerr}{2(1+\epsilon)^2} \ge \frac{\epsilon^2}{2(1+\epsilon)^2} \frac{4}{\epsilon^2}\ln\!\big(\frac{2s}{\delta}\big) \ge \ln\!\big(\frac{2s}{\delta}\big)$ and thus the probability above is bounded by $\frac{\delta}{2s}$.

Suppose instead that, when AGS($\epsilon, \delta$) returns, $\frac{c_i}{w_i} \le g_i(1 - \epsilon)$, i.e., $c_i(1 + \frac{\epsilon}{1-\epsilon}) \le g_i w_i$.
On the one hand this implies that $c_i - g_i w_i \ge \frac{\epsilon}{1-\epsilon} c_i$, that is, upon termination we have $-Z_t \ge \frac{\epsilon}{1-\epsilon}\cerr$.
Obviously $(-Z_t)_{t\ge 0}$ is a martingale too with respect to the filter $(\mathcal{F}_t)_{t\ge 0}$, hence \eqref{eq:martbound} holds if we replace $Z_t$ with $-Z_t$.
Let $t_0 \le t$ be the first step where $-Z_{t_0} \ge \frac{\epsilon}{1-\epsilon} \cerr$; since $|Z_t - Z_{t-1}| \le 1$, it must be $-Z_{t_0} < \frac{\epsilon}{1-\epsilon}\cerr + 1$.
Moreover, $\sum_{\tau=1}^t X_\tau$ is nondecreasing in $t$, so $\sum_{\tau=1}^{t_0} X_\tau \le \cerr$.
It follows that $\sum_{\tau=1}^{t_0} P_\tau = -Z_{t_0} + \sum_{\tau=1}^{t_0} X_\tau < \frac{\epsilon}{1-\epsilon} \cerr + 1 + \cerr = \frac{1}{1-	\epsilon} \cerr +1$.
Applying again \eqref{eq:martbound} with $z=\frac{\epsilon}{1-\epsilon} \cerr$ and $v=\frac{1}{1-\epsilon} \cerr +1$, we obtain:
\begin{align}
\prob\!\left[\frac{c_i}{w_i} \le g_i(1-\epsilon)\right]
&\le \exp{\!\left(-\frac{(\cerr\frac{\epsilon}{1-\epsilon})^2}{2(\frac{1+\epsilon}{1-\epsilon}\cerr+1)}\right)}
\le \exp{\!\left(-\frac{\epsilon^2 \cerr^2}{2(1+\cerr)}\right)}
\end{align}
but since $\cerr \ge 4$ then $\frac{\cerr}{1 + \cerr} \ge \frac{4}{5}$ and so $\frac{\epsilon^2 \cerr^2}{2(1+\cerr)} \ge \frac{2 \epsilon^2 \cerr}{5}$.
By replacing $\cerr$ we get $\frac{2 \epsilon^2 \cerr}{5} \ge \frac{2 \epsilon^2}{5} \frac{4}{\epsilon^2}\ln\!\big(\frac{2s}{\delta}\big) > \ln\!\big(\frac{2s}{\delta}\big)$ and thus once again the probability of deviation is bounded by $\frac{\delta}{2s}$.

By using the union bound on the two cases, the probability that $\frac{c_i}{w_i}$ is not within a factor $(1 \pm \epsilon)$ of $g_i$ is at most $\frac{\delta}{s}$.
Using the union bound once again on all $i \in [s]$, we obtain theorem~\ref{THM:AGS_APX}.

\subsection{Sampling efficiency of AGS}

\subsubsection{Near-optimality of AGS}
Imagine a ``clairvoyant'' algorithm that knows, for every treelet $T_j$, the number of invocations of \sample($T_j$) necessary in order to get the desired accuracy bounds while minimizing the total number of taken samples.
We show that the number of samples used by AGS is close to the number of samples used by this clairvoyant algorithm.
Formally, we prove:
\begin{theorem}
\label{THM:AGS_COST}
If the treelet $T_{j^*}$ chosen by AGS at line~\ref{ags:estim} minimizes $\prob[$\sample$(T_j)$ spans a copy of some $H_i \in C]$, then AGS makes a number of calls to \sample$()$ that is at most $O(\ln(s))=O(k^2)$ times the minimum needed to ensure that every graphlet $H_i$ appears in $\bar c$ samples in expectation.
\end{theorem}
The rest of this section is devoted to proving Theorem~\ref{THM:AGS_COST}.
We will write the problem of minimizing the total number of samples as a covering problem, and show that AGS is a greedy algorithm for that problem.
This will allow us to prove the guarantees of AGS by adapting a standard proof for greedy algorithms for covering problems.

For each $i \in [s]$ and each $j \in [\varsigma]$ let $a_{ji}$ be the probability that \sample($T_j$) returns a copy of $H_i$.
Note that $a_{ji} = g_i\sigma_{ij}/t_j$, which is the fraction of colorful copies of $T_j$ that span a copy of $H_i$.
Our goal is to allocate, for each $T_j$, the number $x_j$ of calls to \sample($T_j$), so that (1) the total number of calls $\sum_{j} x_j$ is minimised and (2) each $H_i$ appears at least $\cerr$ times in expectation.
Formally, let $\mathbf{A} = (a_{ji})^{\transpose}$, so that columns correspond to treelets $T_j$ and rows to graphlets $H_i$, and let $\bfx = (x_1,\ldots,x_{\varsigma}) \in \mathbb{N}^{\varsigma}$.
We obtain the following integer program:
\begin{align*}
\left\{
\begin{array}{l}
\min \, \langle\mathbf{1}, \bfx\rangle\\
\text{s.t.} \, \mathbf{A} \bfx \ge \cerr\,\mathbf{1}\\
\phantom{\text{s.t.}} \, \bfx \in \mathbb{N}^{\varsigma}
\end{array}
\right.
\end{align*}

We now describe the natural greedy algorithm for this problem; it turns out that this is precisely AGS.
The algorithm works in steps.
Let $\bfx^0 = \mathbf{0}$, and for all $t \ge 1$ denote by $\bfx^t$ the partial solution after $t$ steps.
The vector $\mathbf{A} \bfx^t$ is an $s$-entry column whose $i$-th entry is the expected number of occurrences of $H_i$ drawn using the sample allocation given by $\bfx^t$.
We define the vector of residuals at time $t$ as $\bfc^t = \max(\mathbf{0}, \bfc - \mathbf{A} \bfx^t)$, and we write $c^t =  \langle\mathbf{1}, \bfc^t\rangle$.
Note that $\bfc^0 = \cerr\,\mathbf{1}$ and $c^0 = s \cerr$.
Finally, we let $U^t = \{i : c^t_i > 0\}$; this is the set of graphlets not yet covered at time $t$.
Clearly, $U^0 = [s]$.

At step $t$, the algorithm chooses $T_{j^*}$ such that \sample($T_{j^*}$) spans an uncovered graphlet with the highest probability.
To this end, it computes:
\begin{align}
\label{eqn:istar}
j^* := \arg \max_{j=1,\ldots,{\varsigma}} \sum_{i \in U_t} a_{ji}
\end{align}
It then lets $\bfx^{t+1} = \bfx^t + \mathbf{e}_{j^*}$, where $\mathbf{e}_{j^*}$ is the indicator vector of $j^*$, and updates $\bfc^{t+1}$ accordingly.
The algorithm stops when $U^t = \emptyset$, since then $\bfx^t$ is a feasible solution.
We prove:
\begin{lemma}
\label{lem:greedy_cost}
The greedy algorithm returns a solution of cost $O(z\ln(s))$, where $z$ is the cost of the optimal solution.
\end{lemma}
\balance
\begin{proof}
Let $w_j^t = \sum_{i \in U_t} a_{ji}$.
For any $j \in [\varsigma]$ let $\Delta_j^t = c^{t} - c^{t+1}$. This is the decrease in the overall residual weight that we would obtain if $j^* = j$.
Note that $\Delta_j^t \le w_j^t$.
We consider two cases.\\
\textbf{Case 1}: $\Delta_{j^*}^t < w_{j^*}^t$.
This means that for some $i \in U_t$ we have $c_i^{t+1} = 0$, implying $i \notin U_{t+1}$.
In other terms, $H_i$ becomes covered at time ${t+1}$.
Since the algorithm stops when $U_t = \emptyset$, this case occurs at most $|U^0| = s$ times.
\\
\textbf{Case 2}: $\Delta_{j^*}^t = w_{j^*}^t$.
Suppose that the original problem admits a solution with cost $z$.
Obviously, the ``residual'' problem where $\mathbf{c}$ is replaced by $\mathbf{c}^{t}$ admits a solution of cost $z$, too.
This implies the existence of $j \in [\varsigma]$ with $\Delta_j^t \ge \frac{1}{z} c^t$; otherwise, any solution for the residual problem would have cost strictly larger than $z$.
But, by the choice of $j^*$, we have $\Delta_{j^*} = w_{j^*}^t \ge w_{j}^t \ge \Delta_j^t$ for any $j$, hence $\Delta_{j^*}^t \ge \frac{1}{z} c^t$.
Thus by choosing $j^*$ we get $c^{t+1} \le (1-\frac{1}{z})c^t$.
Therefore, after running into this case $\ell$ times, the residual cost is at most $c^0 (1-\frac{1}{z})^\ell$.

Note that $\ell + s \ge c^0 = s \cdot \cerr$ since at any step the overall residual weight can decrease by at most $1$.
Therefore the algorithm performs $\ell + s = O(\ell)$ steps.
Furthermore, after $\ell + s$ steps we have $c^{\ell+s} \le s \cerr e^{-\frac{\ell}{z}}$.
By choosing $\ell=z\ln(2s)$, we obtain $c^{\ell+s} \le \frac{\cerr}{s}$, and therefore each one of the $s$ graphlets receives a weight of at least $\frac{\cerr}{2}$.
To correct the factor $\frac{1}{2}$, replace $\cerr\,\mathbf{1}$ with $ 2 \cerr\,\mathbf{1}$ in the original problem.
The cost of the optimal solution is then at most $2 z$, and in $O(z \ln(s))$ steps the algorithm finds a cover where each graphlet has weight at least $\cerr$.
\end{proof}
Now, note that the treelet index $j^*$ given by~\eqref{eqn:istar} remains unchanged as long as $U_t$ remains unchanged.
Therefore we need to recompute $j^*$ only when some new graphlet exits $U_t$, i.e., becomes covered.
In addition, we do not need each value $a_{ji}$, but only their sum $\sum_{i \in U_t} a_{ji}$.
This is precisely the quantity that AGS estimates at line~\ref{ags:estim}.
Theorem~\ref{THM:AGS_COST} follows immediately as a corollary.

\subsubsection{A general lower bound}
We conclude by showing a lower bound for \emph{all} algorithms based solely on the primitive \sample($T$).
This includes many natural graphlet sampling algorithms such as~\cite{Jha&2015,Wang&2015,Wang&2016}.
\begin{theorem}
\label{thm:samplelb}
For any constant $k \ge 2$ there are graphs on an arbitrarily large number of nodes $n$ such that (i) some graphlet $H$ represents a fraction $p_H = 1/\operatorname{poly}(n) = \Omega(n^{1-k})$ of all graphlet copies, and (ii) any algorithm needs $\Omega(1/p_H)$ calls to \sample$(T)$ in expectation to just find one copy of $H$.
\end{theorem}
\begin{proof}
Let $T$ and $H$ be the path on $k$ nodes.
Let $G$ be the $(n-k+2,k-2)$ lollipop graph; so $G$ is formed by a clique on $n-k+2$ nodes and a dangling path on $k-2$ nodes, connected by an arc.
$G$ contains $\Theta(n^k)$ non-induced occurrences of $T$ in $G$, but only $\Theta(n)$ induced occurrences of $H$ (all those formed by the $k-2$ nodes of the dangling path, the adjacent node of the clique, and any other node in the clique).
Since there are at most $\Theta(n^k)$ graphlets in $G$, then $H$ forms a fraction $p_H=\Theta(n^{1-k})$ of all the graphlets.
Obviously, $T$ is the only spanning tree of $H$.
However, an invocation of $\operatorname{sample}(G,T)$ returns $H$ with probability $\Theta(n^{1-k})$, and therefore we need $\Theta(n^{k-1})=\Theta(1/p_H)$ samples in expectation before obtaining $H$.
One can make $p_H$ larger by considering the $(n',n-n')$ lollipop graph for larger values of $n'$.
\end{proof}

\subsection{Lower-level optimizations and architectural details}
For completeness and reproducibility, as we did for the build-up phase, we describe some lower-level details of uniform sampling and AGS.

\subsubsection{Alias method sampling}
Recall that sampling starts by drawing a node $v$ with probability proportional to $c(\tc,v)$.
We do this in time $O(1)$ by using the alias method~\cite{Vose91}.
This method requires us to build an \emph{alias table}, which requires $O(n)$ time and space.
For uniform sampling, the alias table is built in the second stage of the build-up phase.
For AGS, the alias table is rebuilt every time a new treelet is selected.
In any case we observe that, in practice, building the alias table consumes a negligible fraction of the overall running time.

\subsubsection{Neighbor buffering}
We have observed that, if $G$ has a node $v$ with degree $d_v=\Delta$ much higher than all other nodes, then the sampling rate is very low.
We argue that the reason is the following.
First, if $\Delta$ is large then $c(\tc,v)$ is large, so the sampling routine will often choose $v$ as root node. 
Second, drawing a neighbor of $u$ will take time, as we have to potentially scan $\Theta(\Delta)$ nodes.
The combination of these two effects imply that, if $\Delta$ is large, the sampling phase will spend most of the time scanning the neighbors of $v$.
To mitigate this problem, the first time we sample a neighbor of $v$, we actually sample \emph{many} neighbors of $v$, say $B$, and keep them for later.
To this end we use reservoir sampling, which allows us to sample $B$ neighbors by sweeping over them only once, and thus at the same cost of sampling just one neighbor.
We apply this trick to every node $v$ with $d_v \ge \Delta_0$ for some tunable parameter $\Delta_0$.
As a result, we scan the neighbors of large-degree nodes only once in a while.
As Figure~\ref{fig:buffered} shows, this increases the sampling speed of \motivo\ significantly (we note that \motiveight\ already achieves those sampling rates and buffering does not increase it further).
\begin{figure}[ht]
\centering
\includegraphics[scale=.85]{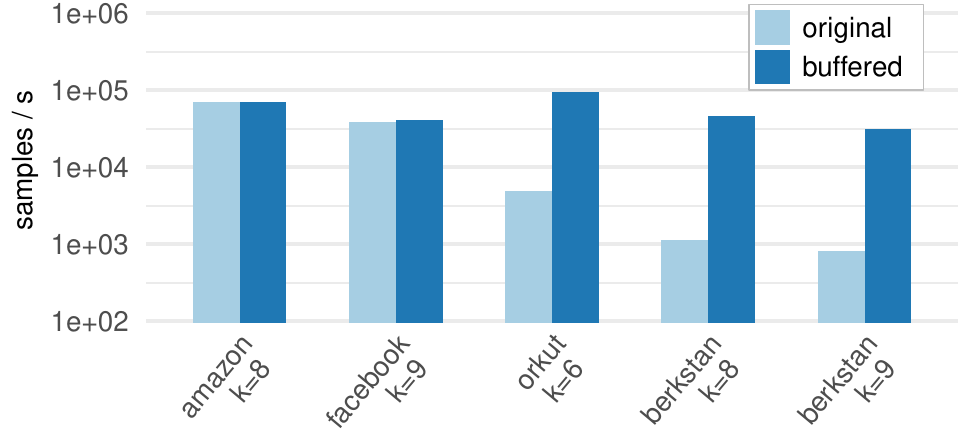}
\caption{\normalfont impact of neighbor buffering on sampling.}
\label{fig:buffered}
\end{figure}

\subsubsection{Graphlet manipulations}
Recall that, after sampling a graphlet occurrence, we have to perform isomorphism tests (to identify its class $H$) and compute its spanning trees (in order to weigh the sample).
To perform the isomorphism test, we first replace the graphlet with a canonical representative of its isomorphism class, computed using the Nauty library~\cite{McKay201494}.
Then, as the $k \times k$ (boolean) adjacency matrix of the graphlet is symmetric with diagonal $0$, we pack it as a $(k-1) \times \frac{k}{2}$ matrix if $k$ is even and in a $k \times \frac{k-1}{2}$ matrix if $k$ is odd (see e.g.\ \cite{BaroudiSL17}). 
Finally, we further reshape this matrix into a $1 \times \frac{k^2 - k}{2}$ vector, which fits into $128$ bits for all $k \le 16$.
Therefore, every graphlet is mapped into a $128$-bit string identifying its isomorphism class, and the isomorphism test between two graphlets boils down to comparing these $128$-bit encodings.

To compute the number of spanning trees $\sigma_i$ of $H_i$, we employ Kirchhoff's matrix-tree theorem, which relates $\sigma_i$ to the determinant of a submatrix of the Laplacian $H_i$.
The running time is $O(k^3)$.
To compute the number $\sigma_{ij}$ of occurrences of a specific treelet $T_i$ in $H_j$ (needed for our sampling algorithm AGS, see Section~\ref{sec:ags}), we use an in-memory implementation of the build-up phase where each vertex $H_j$ is assigned a distinct color in $\{0, \dots, k-1\}$.

\section{Experimental results}
\label{sec:exp}
We measure the performance of \motivo\ and \motiveight\ in terms of running time, space usage, and accuracy of the counts, with a special attention towards \motiveight.
We recall that CC is the current state of the art; in particular, algorithms based on random walks are outperformed by CC, see~\cite{Bressan&2018b}.
Therefore, we compare only against CC.
All our experiments are performed on an Amazon EC2 \texttt{c5d.9xlarge} instance, with 36 virtual CPUs, 72GB of main memory, and a 900GB solid-state disk drive used to store the count tables.

To begin, we tested \motivo\ and \motiveight\ on all our graphs for increasing values of $k$, stopping when witnessing a slowdown due to excessive I/O (recall that our algorithms must repeatedly read and store the count tables on disk).
In the sampling phase, we took $5$ million samples. As we show below, this was sufficient to guarantee high accuracy on most graphlets.
Table~\ref{tab:graphs} summarizes the results.
Using \motiveight\ we reached $k=8$ on all our graphs, and using \motivo\ we reached $k>8$ on half of them.
The table does not show the \yeast\ graph~\cite{jeong2001lethality}, a small graph on which we successfully ran \motivo\ for $k=16$ in less than three hours (we recall that there are $6 \cdot 10^{22}$ distinct motifs on $16$ nodes).
For comparison,~\cite{Alon&2008} on \yeast\ reached $k=10$ and only on tree-like graphlets.
\begin{table}[ht]
\caption{Summary of our results. For each graph we report the maximum reached value of $k$ and the total wall time ($^*$ = with biased coloring). The wall time includes sampling.}
\label{tab:graphs}
{
\begin{tabular}{lrrllll}
graph & nodes (millions) & edges (millions) & {source} & {k} & wall time & algorithm \\ 
\hline
\facebook & $0.1$ & $0.8$ & \SWS & 11 & 1h & \motivo\\
\dblp & $0.9$ & $3.4$ & \SNAP & 9 & 7m & \motivo\\
\amazon & $0.7$ & $3.5$ & \SNAP & 9 & 8m & \motivo\\
\bstan & $0.7$ & $6.6$ &  \SNAP & 9 & 55m &\motivo\\
\yelp & $7.2$ & $26.1$ & \YLP & 8 & 13m &\motiveight\\
\lj & $5.4$ & $49.5$ & \LAW & 8 & 24m &\motiveight\\
\orkut & $3.1$ & $117.2$ & \SWS & 8 & 1h11m &\motiveight\\
\twitter & $41.7$ & $1202.5$ & \LAW & 8$^*$ & 2h45m&\motiveight\\
\fs &  $65.6$ & $1806.1$ & \SNAP & 8$^*$ & 1h10m & \motiveight
\end{tabular}
}
\end{table}

\subsection{Computational efficiency}
\label{sub:perf}
Figure~\ref{fig:abstime} shows the performance of \motivo\ and \motiveight\ for $k=8$: the running time (seconds), the total space usage of the build-up phase (GB), and the speed of uniform sampling (graphlets/second).
This shows how \motiveight\ manages graphs significantly larger than the state of the art.
Note also the reduction in space usage of \motiveight\ compared to \motivo.
For both \twitter{} and \fs{}, we used biased coloring to keep the build-up time below 3 hours.
\begin{figure}[h!]
\centering
\includegraphics[scale=.71]{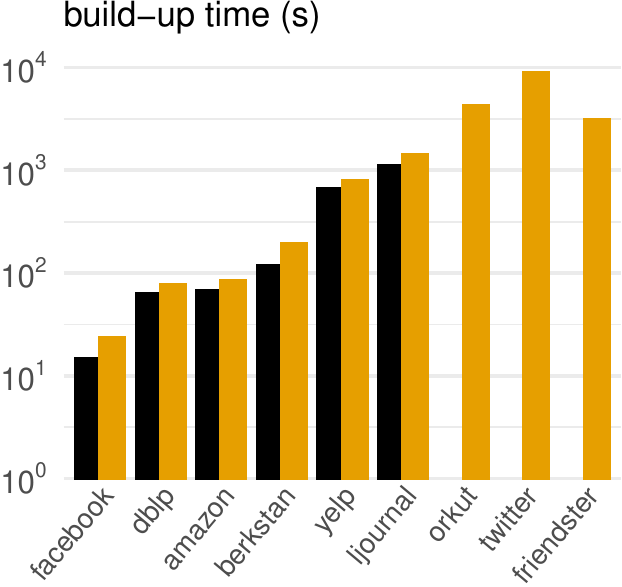}
\hspace*{10pt}
\includegraphics[scale=.71]{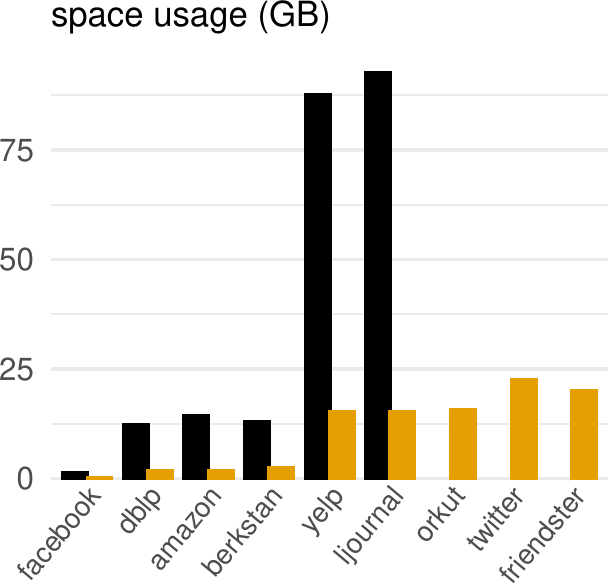}
\hspace*{10pt}
\includegraphics[scale=.71]{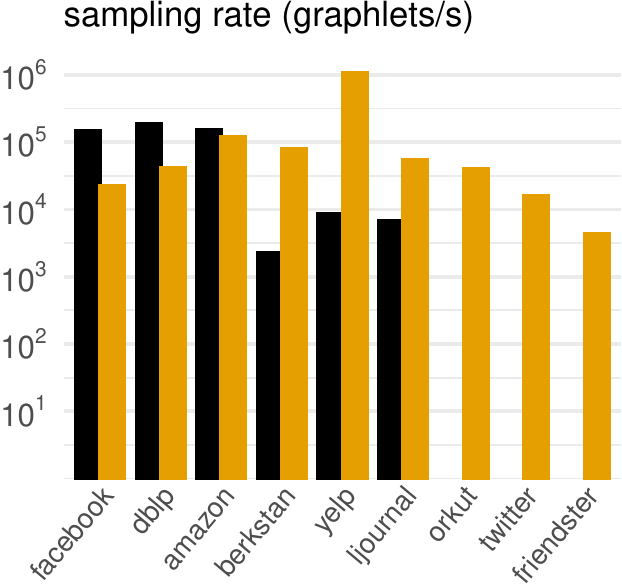}
\caption{Computational performance of \motivo\ (black bars) and \motiveight\ (yellow bars), for $k=8$.
Missing bars represent the failure of \motivo\ by memory exhaustion.
Note the drop in space usage of \motiveight\, which allows it to run even where \motivo\ fails.}
\label{fig:abstime}
\end{figure}

We also measured the performance of \motiveight\ as a function of $n=|V(G)|$ and $m=|E(G)|$, by computing the average time \emph{per million edges} and the average space \emph{per node} on all our graphs, see Figure~\ref{fig:perf}.
We did this to show that \motiveight\ is predictable as a function of $n$ and $m$.
This sets it apart from most graphlet counting algorithms, whose running time varies chaotically.
For instance, ESCAPE takes $5$ seconds on a graph with $1.2$M edges and $11$ days on a graph with $3.6$M edges, a blow-up of $175.000$ times~\cite{Pinar&2017}.
The reason is that the algorithm enumerates all occurrences in $G$ of certain ``critical'' subgraphs (for instance, cliques), whose number can vary wildly between graphs of comparable size.
A similar chaotic behaviour is exhibited by random walks~\cite{Bressan&2017,Bressan&2018b}, since their mixing time depends on the ratio between the maximum and minimum degree of $G$, raised to the power of $\Theta(k)$~\cite{Agostini&19}.
\begin{figure}[h!]
\centering
\hfill
\includegraphics[width=.47\textwidth]{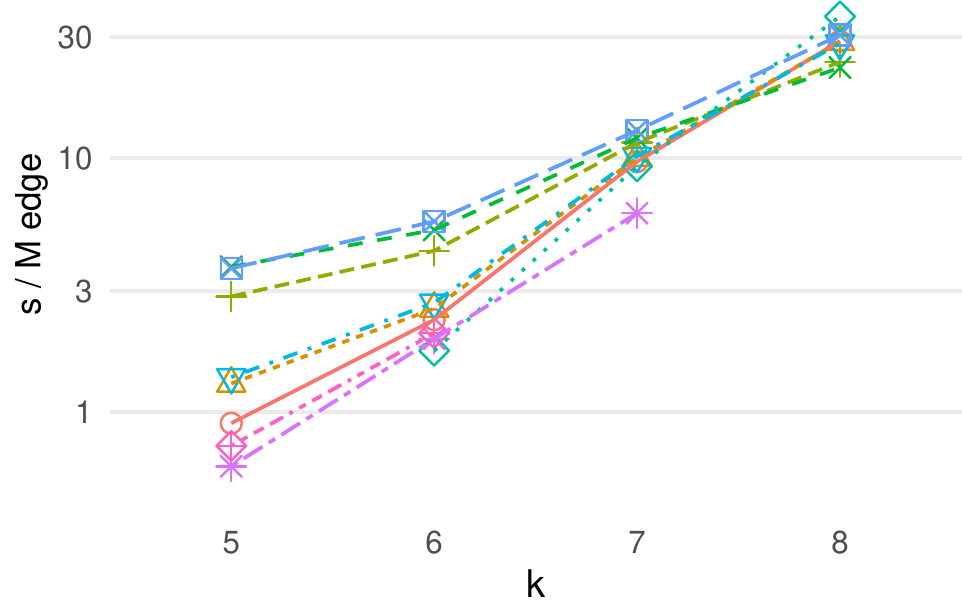}
\hfill
\includegraphics[width=.47\textwidth]{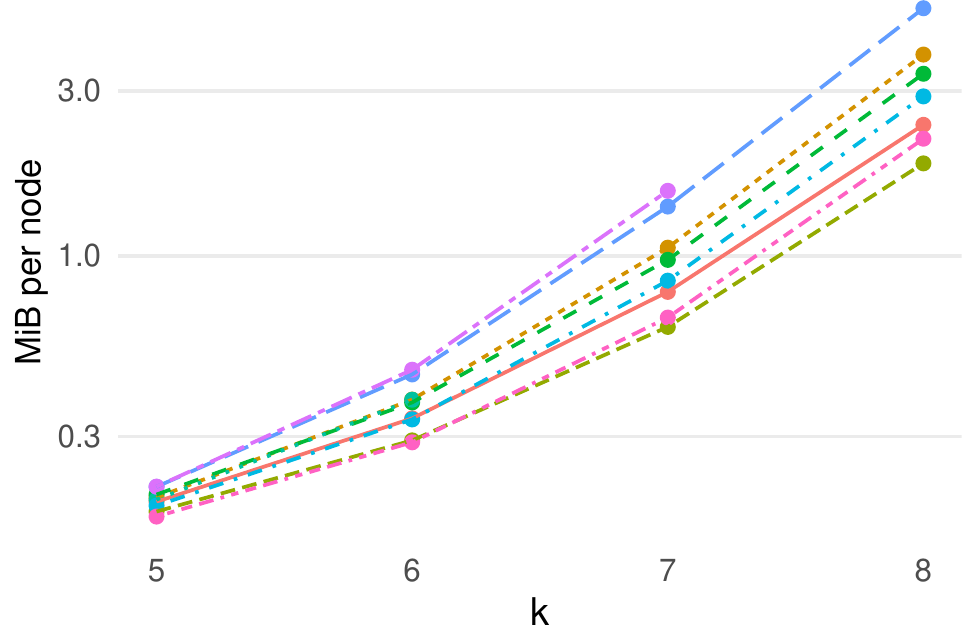}
\hfill\phantom{.}
\caption{Build-up time in seconds per million edge, and space usage in bits per input node, on all our graphs, showing how \motiveight\ is predictable as a function of $n$, $m$, and $k$.
}
\label{fig:perf}
\end{figure}

\textbf{Comparison against CC.}
We compare \motivo\ and \motiveight\ against CC in Table~\ref{tab:speedup} and Table~\ref{tab:speedup_2}.
For each graph we report the largest $k$ for which CC ran, without dying by memory exhaustion or integer overflow.
For the space usage, we compare the main memory used by CC to the total external memory usage of our algorithms (recall that CC works in main memory).
The sampling rate refers to uniform sampling, which is the only one supported by CC.
The sampling speed of AGS is similar, and never more than $40\%$ lower than uniform sampling.
The only exception is the \yelp\ graph, on which AGS for $k=8$ is $20\times$ slower than uniform sampling.
But it is also much more accurate and still $10\,000 \times$ faster than CC, as we show below.

\begin{table}[ht]
\resizebox{\textwidth}{!}{%
\begin{tabular}{llrrrrrrrrr}
graph & k & \multicolumn{3}{l}{\small build-up time (seconds)} &
\multicolumn{3}{l}{\small build-up space (GB)} &
\multicolumn{3}{l}{\small sampling rate (motifs/sec) } \\ & & CC & \motivo\ & speedup & CC & \motivo\ & reduction & CC & \motivo\ & speedup
\\\midrule
\facebook\ & 9 & 860 & 86 & 10$\times$ & 33 & 5 & 7$\times$ & 5& 70k & 14\,000$\times$\\
\facebook\ & 8 & 95 & 17 & $5.5\times$ & 24 & 1.4 & $17\times$ & 419& 149k & 356$\times$\\
\dblp\ & 9 & 1245 & 320 & 3.9$\times$ & 43 & 44 & 1$\times$ & 72 & 112k& 1\,556$\times$\\
\dblp\ & 8 & 182 & 82 & $2.2\times$ & 30 & 27 & $1\times$ & 679 & 186k& 274$\times$\\
\amazon\ & 9 & 376 & 84 & 2.2$\times$ & 49 & 51 & 1$\times$ & 226 & 87k & 385$\times$\\
\amazon\ & 8 & 140 & 87  & $1.6\times$ & 33 & 31 & $1\times$ & 1556 & 150k & 96$\times$\\
\bstan\ & 5 & 14 & 7 & 2$\times$ & 18 & 0.5 & 36$\times$ & 160 & 6770 & 42$\times$\\
\yelp\ & 5 & 167 & 71 & 2.4$\times$ & 36 & 4.5 & 8$\times$ & 20 & 910 & 45$\times$\\
\lj\ & 6 & 306 & 99 & 3$\times$ & 36 & 9.5 & 4$\times$ & 295 & 12k & 41$\times$\\
\orkut\ & 5 & 225 & 40 & 5.6$\times$ & 27 & 3.2 & 8$\times$ & 295 & 15k & 51$\times$
\end{tabular}
}
\caption{Computational performance of \motivo\ versus CC.}
\label{tab:speedup}
\end{table}

\begin{table}[ht]
\resizebox{.9\textwidth}{!}{%
\begin{tabular}{llrrrrrrrrr}
graph & k & \multicolumn{3}{l}{\small build-up time (seconds)} &
\multicolumn{3}{l}{\small build-up space (GB)} &
\multicolumn{3}{l}{\small sampling rate (motifs/sec) } \\ & & CC & LM & speedup & CC & LM & reduction & CC & LM & speedup
\\\midrule
\facebook\ & 8 & 95 & 22 & 4.3$\times$ & 24 & 0.2 & 114$\times$ & 420& 23k & 55$\times$\\
\dblp\ & 8 & 182 & 77 & 2.4$\times$ & 30 & 1.7 & 1.8$\times$ & 680& 41k& 60$\times$\\
\amazon\ & 8 & 140 & 84 & 1.7$\times$ & 31 & 1.8 & 17$\times$ & 1550 & 120k & 77$\times$\\
\bstan\ & 5 & 14 & 7 & 2$\times$ & 18 & 0.2 & 130$\times$ & 160 & 1.2M & 7\,700$\times$\\
\yelp\ & 5 & 167 & 69 & 2.4$\times$ & 36 & 1.3 & 27$\times$ & 20 & 4.2M & 210\,000$\times$\\
\lj\ & 6 & 306 & 79 & 3.9$\times$ & 36 & 1.8 & 20$\times$ & 295 & 112k & 380$\times$\\
\orkut\ & 5 & 225 & 38 & 6$\times$ & 27 & 0.7 & 40$\times$ & 295 & 162k & 550$\times$
\end{tabular}
}
\raggedright
\caption{Computational performance of \motiveight\ versus CC. Note that the maximum value of $k$ tested is $k=8$ as this is the largest value supported by \motiveight.}
\label{tab:speedup_2}
\end{table}

\subsection{Accuracy of the estimates and performance of AGS}
\label{sub:accuracy}
We evaluate the accuracy of the estimates produced by our algorithm \motiveight.

\paragraph{Uniform sampling.}
First, we consider uniform sampling.
In this case, the accuracy of \motiveight\ and the accuracy of \motivo\ are identical, as they give the same output (only the sampling rate changes).
We started by computing the ground-truth count of the number of copies of each possible $k$-graphlet in each graph.
For $k=5$ we used the exact algorithm ESCAPE~\cite{Pinar&2017} which works well on small graphs (\facebook, \dblp, \amazon, \lj, \orkut).
On all other graphs and/or for $k > 6$, we used as ground truth the average of the counts returned by $20$ independent runs of \motiveight, of which $10$ used uniform sampling and $10$ used AGS.
We then measured the average accuracy of \motiveight\ against the ground truth over $10$ runs.
In each run we took $10$M samples, or stopped the sampling after $600$s ($10$ minutes).
To measure the accuracy, we use the relative error.
We denote by $c_H$ the ground-truth count of a specific graphlet $H$, and let $\hat{c}_H$ be the estimate returned by the algorithm.
Then we define the relative error of $H$ as:
\begin{align}
\text{err}_H = \frac{\hat{c}_H - c_H}{c_H}.
\end{align}
Therefore a value $\text{err}_H \simeq 0$ means $c_H$ has been accurately estimated; whereas $\text{err}_H \gg 0$ means $c_H$ is overestimated, and $\text{err}_H \ll 0$ means $c_H$ is underestimated (and $\text{err}_H=-1$ means no copy of $H$ was found).

Figure~\ref{fig:errdist} shows the distribution of the relative error $\text{err}_H$ for uniform sampling, for $k=7$, on three representative graphs.
The $x$-axis shows the value of $\text{err}_H$, and the $y$-axis the number of distinct graphlets $H$ for which that value is achieved.
\begin{figure*}[h!]
\centering
\includegraphics[width=.9\textwidth]{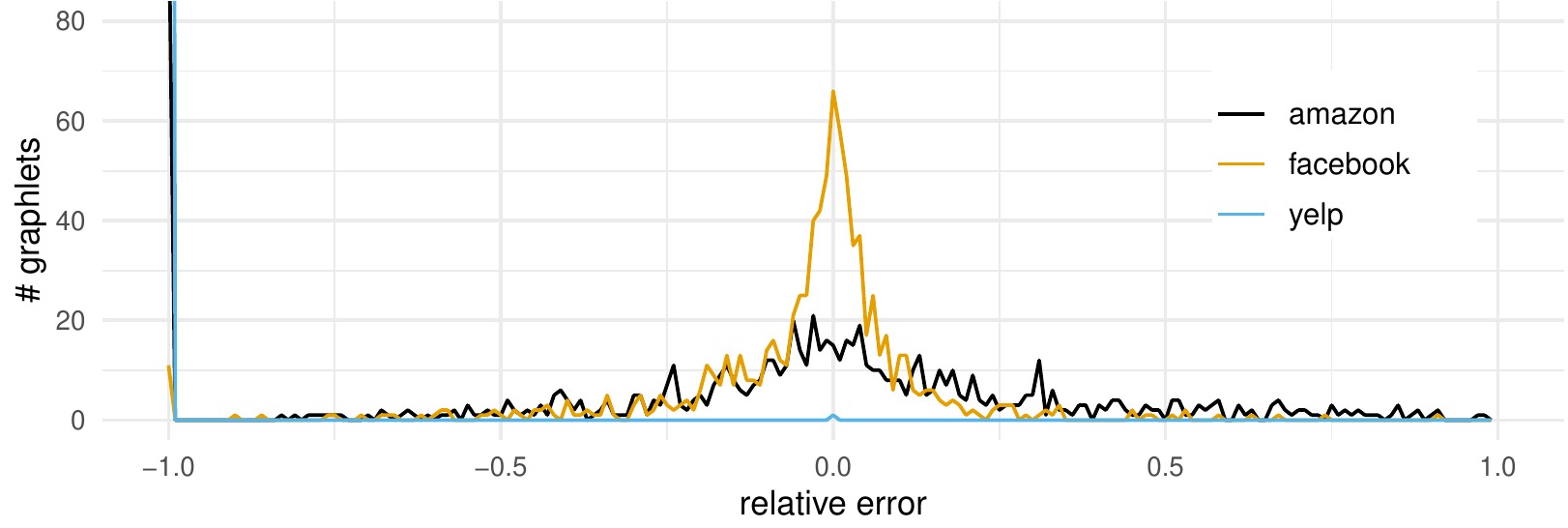}
\caption{Relative error distribution of uniform sampling for $k=7$. On \yelp\ and \amazon\ almost all graphlets have a relative error of $-1$, i.e., they are completely missed.}
\label{fig:errdist}
\end{figure*}
Note that for \yelp\ and \amazon\ almost all graphlets have $\text{err}_H=-1$, as can be seen by the straight segments leaving in the leftmost part of the plot.
This means uniform sampling misses almost all graphlets on \amazon\ and \yelp.

\paragraph{AGS.}
Figure~\ref{fig:errdist2} gives the relative error distribution for AGS.
Note how the distribution is now concentrated around $0$.
This means that AGS gives an accurate estimate of nearly all graphlets, in line with our theoretical predictions.
\begin{figure*}[h!]
\centering
\includegraphics[width=.9\textwidth]{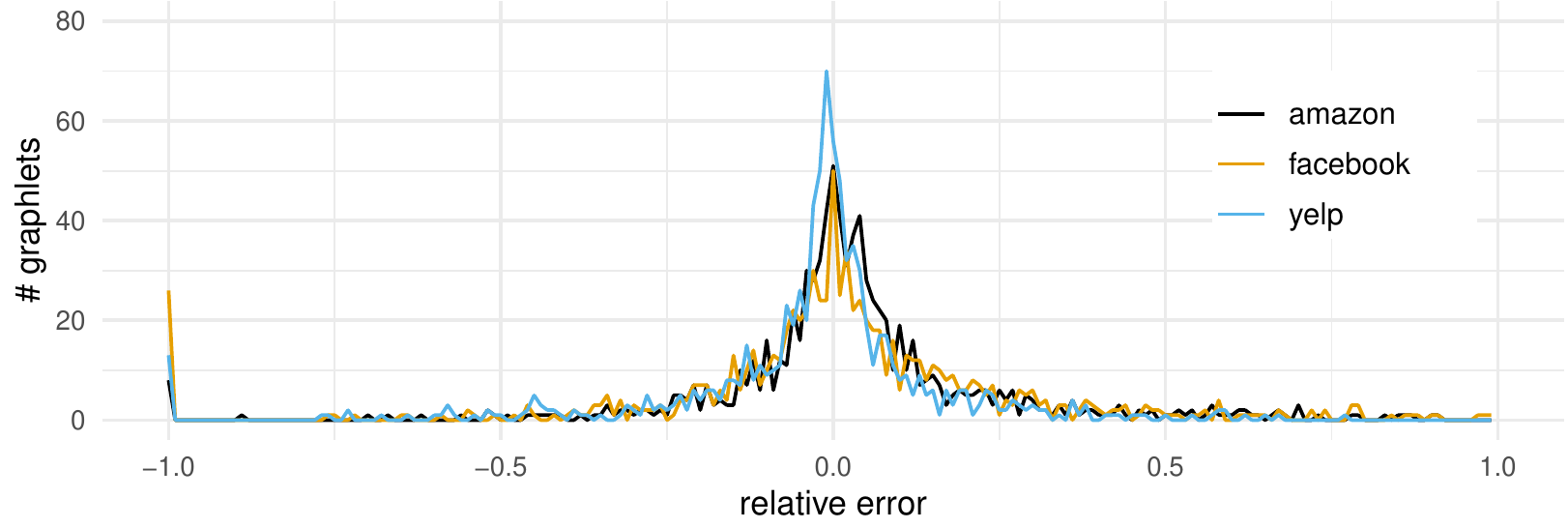}
\caption{Relative error distribution of AGS for $k=7$. Unlike the case of uniform sampling, here almost all graphlets are accurately estimated and have a relative error close to $0$.}
\label{fig:errdist2}
\end{figure*}

To complete the evaluation of AGS, we computed the number of graphlets with relative error below $0.25$.
This number is shown in Figure~\ref{fig:ok}, where the shaded area represents the maximum achievable, i.e., $N_8=11\,117$ (the number of non-isomorphic simple connected graphs on $8$ nodes).
The plot is particularly telling if we look at the \yelp\ graph.
According to our ground truth, in this graph over $99.9996\%$ of all $8$-graphlets are stars.
Thus, we can expect uniform sampling to waste essentially all of its samples by drawing stars.
The figure shows this is exactly the case.
Indeed, uniform sampling achieves a relative error $\le 25\%$ only for the $4$ most frequent graphlets (as a fraction, 0.04\% of the total).
AGS instead achieves a relative error $\le 0.25$ for 9\,860 graphlets (as a fraction, 89\% of the total).
This includes many graphlets with frequency below $10^{-21}$.
(These graphlet are well-estimated in all $10$ runs, so they are not the product of chance).
To find those graphlets, uniform sampling would need more than $10^3$ years even if running at $10^9$ samples per second.
Therefore, we can say that AGS can count graphlets that uniform sampling cannot.

\begin{figure*}[ht]
\centering
\includegraphics[width=.45\textwidth]{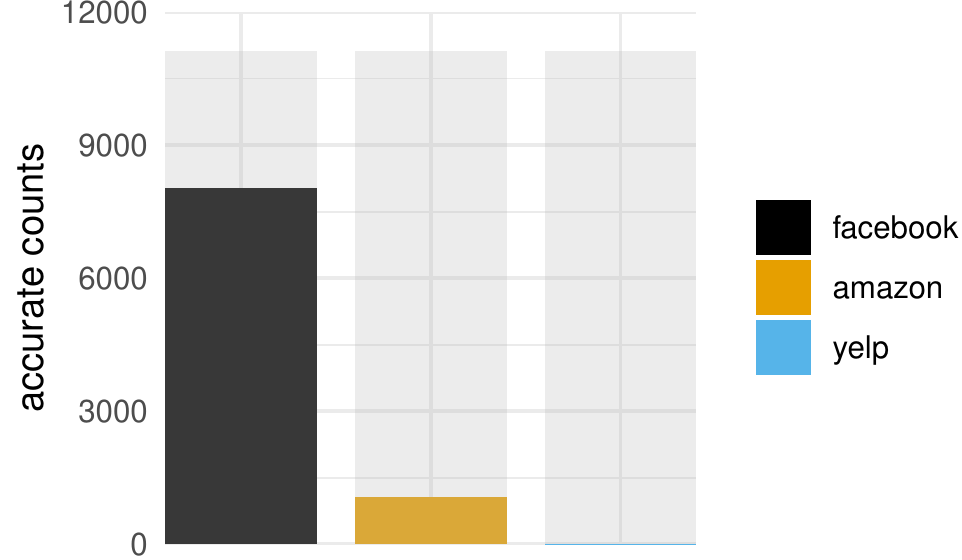}
\hfill
\includegraphics[width=.45\textwidth]{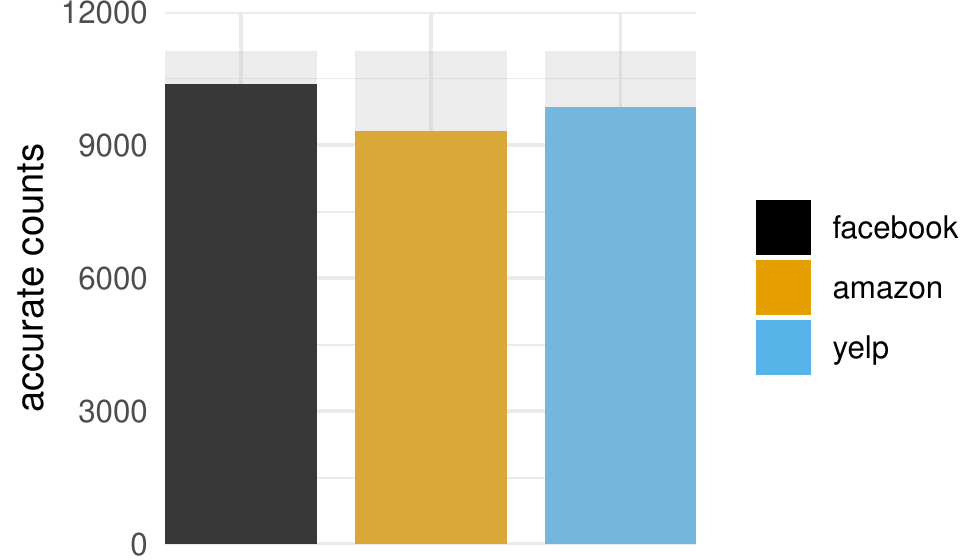}
\caption{Number of $8$-graphlets for which \motiveight\ achieved relative error below $25\%$. Left: uniform sampling. Right: AGS. The shaded area shows the total number of $8$-graphlets, $N_8=11\,117$.}
\label{fig:ok}
\end{figure*}

\section{Conclusions}
In this work we confirm that color coding is an effective technique for sampling and counting motifs in large graphs.
Although this was already suggested by existing work, here we refine the approach and push the color coding motif mining paradigm forward.
It would be interesting to investigate how this color coding approach could be extended to richer and more challenging scenarios. 
Two of these scenarios that fit well with the assumption of large graphs are a distributed computing setting and graphs that evolve in time.

\bibliographystyle{abbrv}
\bibliography{biblio}


\end{document}